\RequirePackage{fix-cm}
\documentclass[smallextended,numbook]{svjour3mod}       %
\smartqed  
\usepackage{graphicx}
\usepackage[utf8]{inputenc}
\usepackage{amsmath}
\usepackage{amssymb}
\usepackage{cite}
\usepackage{bbm,mathtools}
\usepackage[dvipsnames]{xcolor}
\usepackage{soul}
\usepackage{authblk,url}

\newcommand{\BT}{\mathrm{BT}}
\newcommand{\difd}{\mathrm{d}}
\newcommand{\expe}{\mathrm{e}}
\newcommand{\bzero}{\bs{0}}
\newcommand{\mprior}{\bs{\mu}_\mathrm{pr}}
\newcommand{\mpos}{\bs{\mu}_\mathrm{pos}}
\newcommand{\mposBT}{\bs{\mu}_{\mathrm{pos},\mathrm{BT}}}
\newcommand{\hmpos}{\widehat{\bs{\mu}}_\mathrm{pos}}
\newcommand{\gmeas}{\bs{\Gamma}_{\bs{\epsilon}}}
\newcommand{\gmeasi}{\bs{\Gamma}_{\bs{\epsilon}_i}}

\newcommand{\gmi}{\bs{\Gamma}_{\bs{\epsilon}_i}}
\newcommand{\gprior}{\bs{\Gamma}_{\!\mathrm{pr}}}
\newcommand{\gobs}{\bs{\Gamma}_{\!\mathrm{obs}}}
\newcommand{\gobsdyn}{\bs{\Gamma}_{\mathrm{tot}}}

\newcommand{\gpos}{\bs{\Gamma}_{\!\mathrm{pos}}}

\newcommand{\gposBT}{\bs{\Gamma}_{\!\mathrm{pos},\mathrm{BT}}}
\newcommand{\hgpos}{
\widehat{\bs{\Gamma}}_{\!\mathrm{pos}}}

\newcommand{\bPinf}{\bs{P}_\infty}

\usepackage[left=1.25in,right=1.25in,top=1in]{geometry}
\usepackage[footnotesize,bf]{caption}
\usepackage{cleveref}

\DeclareMathOperator*{\argmin}{arg\,min}

\newcommand{\R}{\mathbbm{R}}
\newcommand{\E}{\mathbbm{E}}

\newcommand{\cov}{\mathsf{cov}}

\newcommand{\diag}{\mathsf{diag}}

\newcommand{\dv}[3][]{\frac{\mathrm{d}^{#1}#2}{\mathrm{d}#3^{#1}}}

\newcommand{\bm}{\bs{m}}

\newcommand{\bp}{\bs{p}}

\newcommand{\bt}{\bs{t}}
\newcommand{\bu}{\bs{u}} 
\newcommand{\bv}{\bs{v}}
\newcommand{\bw}{\bs{w}} 
\newcommand{\bx}{\bs{x}}  
\newcommand{\by}{\bs{y}} 
\newcommand{\bz}{\bs{z}}  

\newcommand{\bA}{\bs{A}}
\newcommand{\bB}{\bs{B}}
\newcommand{\bC}{\bs{C}}
\newcommand{\bD}{\bs{D}}

\newcommand{\bG}{\bs{G}}
\newcommand{\bH}{\bs{H}}
\newcommand{\bI}{\bs{I}} 

\newcommand{\bK}{\bs{K}}
\newcommand{\bL}{\bs{L}}
\newcommand{\bM}{\bs{M}}
\newcommand{\bN}{\bs{N}} 
 
\newcommand{\bP}{\bs{P}} 
\newcommand{\bQ}{\bs{Q}} 
\newcommand{\bR}{\bs{R}} 
\newcommand{\bS}{\bs{S}} 
\newcommand{\bT}{\bs{T}}
\newcommand{\bU}{\bs{U}} 
\newcommand{\bV}{\bs{V}}
\newcommand{\bW}{\bs{W}}

\newcommand{\bZ}{\bs{Z}}  
\newcommand{\beps}{\bs{\epsilon}}
\newcommand{\bXi}{\bs{\Xi}}
\newcommand{\bxi}{\bs{\xi}}
\newcommand{\bGamma}{\bs{\Gamma}}

\newcommand{\bSigma}{\bs{\Sigma}}
\newcommand{\bDelta}{\bs{\Delta}}

\newcommand{\cN}{\mathcal{N}}

\newcommand{\nin}{d_{\rm in}}  % input dimension
\newcommand{\nout}{d_{\text{out}}} % output/observation dimension of inference problem (y = Cx)
\newcommand{\nmeas}{n}% number of measurements
\newcommand{\nobs}{d_{\rm obs}} % observation dimension (collection of measurements)
\newcommand{\ns}{d}           % state dimension
\newcommand{\bxin}{\bx_0}
\newcommand{\tstart}{t_{\mbox{\tiny{start}}}}
\newcommand{\tend}{t_{\mbox{\tiny{end}}}}

\newcommand{\Qm}{\bQ_{\bm}}

\newcommand{\Qy}{\bQ_{\by}}

\newcommand{\Pm}{\mathbb{P}} 
\newcommand{\Mset}{M}
\newcommand{\bmu}{\bs{\mu}}
\newcommand{\Cov}{\Sigma}
\DeclareMathOperator*{\argmax}{argmax}

\usepackage{array,multirow,booktabs,bbm} % table stuff

\newcommand{\bs}[1]{\boldsymbol{#1}}

\newcommand{\rev}[1]{{\leavevmode\color{blue}{#1}}}                     % For new/changed text
\begin{document}

\title{Model Reduction of Linear Dynamical Systems\\
via Balancing for Bayesian Inference\thanks{This material is based upon work supported by the National Science Foundation under Grant No. DMS-1439786 and by the Simons Foundation Grant  No. 50736 while the authors were in residence at the Institute for Computational and Experimental Research in Mathematics in Providence, RI, during the ``Model and dimension reduction in uncertain and dynamic systems" program. EQ was supported in part by the Fannie and John Hertz Foundation. JMT was partially supported by EPSRC grant EP/S027785/1. AN was partially supported by NSF DMS-1848508. CB and SG were partially supported  by NSF DMS-1819110.}}

\author{Elizabeth Qian \and
        Jemima M. Tabeart \and 
        Christopher~Beattie \and
        Serkan Gugercin \and
        Jiahua Jiang \and
        Peter R. Kramer \and
        Akil Narayan  \and 
}

\authorrunning{Qian, Tabeart, Beattie, Gugercin, Jiang, Kramer, Narayan}

\institute{E. Qian \at
              California Institute of Technology \and
           J.M. Tabeart \at
              The University of Edinburgh \and
           C. Beattie \at
              Virginia Tech \and
            S. Gugercin \at
              Virginia Tech \and
           J. Jiang \at
              ShanghaiTech University \and
           P.R. Kramer \at
              Rensselaer Polytechnic Institute \and
           A. Narayan \at
              University of Utah \and
}

\date{}

\maketitle

\begin{abstract}
    We consider the Bayesian approach to the linear Gaussian inference problem of inferring the initial condition of a linear dynamical system from noisy output measurements taken after the initial time. In practical applications, the large dimension of the dynamical system state poses a computational obstacle to computing the exact posterior distribution. \emph{Model reduction} offers a variety of computational tools that seek to reduce this computational burden. In particular, \emph{balanced truncation} is a system-theoretic approach to model reduction which obtains an efficient reduced-dimension dynamical system by projecting the system operators onto state directions which trade off the reachability and observability of state directions as expressed through the associated Gramians.  We introduce Gramian definitions relevant to the inference setting and propose a balanced truncation approach based on these inference Gramians that yield a reduced dynamical system that can be used to cheaply approximate the posterior mean and covariance. Our definitions exploit natural connections between (i) the reachability Gramian and the prior covariance and (ii) the observability Gramian and the Fisher information. The resulting reduced model then inherits stability properties and error bounds from system theoretic considerations, and in some settings yields an optimal posterior covariance approximation. Numerical demonstrations on two benchmark problems in model reduction show that our method can yield near-optimal posterior covariance approximations with order-of-magnitude state dimension reduction.
\end{abstract}

\section{Introduction}
The task of updating computational models to take into account empirical observations --- 
a task that plays a central role in \emph{data assimilation} --- occurs
in a diverse range of disciplines, including geophysics, oceanography, and atmospheric sciences \cite{kalnay_atmospheric_2002,bennett_inverse_2005,palacios2006non,evensen_data_2009}. Specific applications include sea ice modeling \cite{rounce_quantifying_2020,yang2020tipping}, geothermal reservoir model calibration \cite{cui_bayesian_2011,cui2019posteriori}, and nitrogen concentration of water bodies \cite{de2002bayesian}.
In the Bayesian inferential approach to data assimilation, model parameters $\bp\in\R^{\ns}$ are related to observed data $\bm\in\R^{\nobs}$ through a measurement model that combines a \emph{forward map} $\bG:\mathbb{R}^{\ns}\rightarrow \mathbb{R}^{\nobs}$ with additive stochastic noise $\beps\in\R^{\nobs}$:
\begin{align}\label{eq:linear-model}
  \bm = \bG(\bp)  + \beps.
\end{align}
Before observations are obtained, uncertainty about the model parameters is encoded in a \emph{prior} probability distribution for $\bp$. The probability distribution of the measurements conditioned on the parameter, $\bm|\bp$, is called the \emph{likelihood}. After observations are obtained, the Bayesian approach computes an updated \emph{posterior} distribution for $\bp|\bm$ which reflects the uncertainty in the parameters conditioned on the measured data and is proportional to the product of the prior and the likelihood of the data.
Once the posterior distribution has been characterized, subsequent steps in application may require computing and utilizing posterior statistics.
 In the linear Gaussian case, this requires computing only a posterior mean and covariance, but for general non-Gaussian posteriors, these statistics are approximated using sampling procedures \cite{soize2020sampling} (e.g., rejection sampling \cite{vonNeumann1951,deligiannidis2020ensemble,devroye1986general} or Markov chain Monte Carlo methods \cite{miroshnikov2015parallel,apte2007sampling,hensman2015mcmc}).
Modern overviews of inference techniques and algorithms can be found in excellent survey articles \cite{stuart_inverse_2010,dashti_bayesian_2017} and textbooks \cite{tarantola_inverse_2005,calvetti_introduction_2007,law_data_2015}. 

In practical settings, several factors may pose computational challenges to the use of Bayesian inference. First, the parameter dimension $\ns$ is often large; e.g., when $\bp$ represents the high-dimensional spatial discretization of a continuous field. Second, it is common to be interested in parameter values, $\bp$, at an initial time, with observations, $\bm$, taken after the initial time. In these settings, the forward model may involve simulation of a high-dimensional dynamical system, making the characterization of the posterior distribution computationally expensive.

To reduce this computational burden, one class of methods exploits the fact that in practice the data are only informative in a low-dimensional subspace of $\mathbb{R}^{\ns}$. 
We highlight the approach of~\cite{spantini_optimal_2015} for the linear Gaussian setting, which defines the posterior covariance as a negative semi-definite update to the prior covariance. The work~\cite{spantini_optimal_2015} shows that the optimal low-rank approximation to this update lies in a subspace spanned by directions identified from dominant eigenvalue-eigenvector pairs of a matrix pencil that encodes a trade-off between the prior uncertainty and the Fisher information of the measurement model. This low-rank approximation performs well when $\nobs\ll\ns$, as is common in data assimilation problems for weather and climate \cite{kalnay_atmospheric_2002,bennett_inverse_2005}. More broadly, this approach is beneficial in cases where the eigenvalues of the matrix pencil exhibit rapid decay, which can occur even if $\nobs>\ns$: some examples are discussed in the introduction of~\cite{spantini_optimal_2015}.  Extensions of the approach in~\cite{spantini_optimal_2015} have been used to develop low rank approximation strategies for inverse problems with nonlinear forward models~\cite{cui2014likelihood}, goal-oriented notions of optimality~\cite{spantini2017goal}, and low-rank tensor-based approaches to reducing computational complexity in time as well as state~\cite{benner2018compression}.

A complementary class of methods employs \emph{model reduction}, which seeks to reduce the expense of evaluating the forward model $\bG$. In projection-based model reduction, a high-dimensional linear dynamical system is approximated by projecting the system operators onto a low-dimensional subspace, yielding an inexpensive low-dimensional dynamical system~\cite{gugercin2004survey,antoulas2005approximation,benner_survey_2015,antoulas_interpolatory_2020}.  One popular strategy known as \emph{balanced truncation}~\cite{moore1981principal,mullis1976synthesis} defines the projection subspace as the span of the dominant eigenvectors of a matrix pencil that encodes a trade-off between states that the system can easily reach and states that contribute more to the system output.  States are retained if they are both highly observable and easily reachable; otherwise, they are discarded.  The resulting reduced-order dynamics evolve in a way that approximate the input-output map of the original large-order dynamical system with high fidelity, frequently with far smaller system order \cite{gugercin2004survey,lieberman2013goal,benner2017balanced,petreczky2012theoretical}.

\subsection{Contributions of this article}
This work considers the Bayesian inference setting where the parameter $\bp$ to be inferred is the initial state (at time $t = 0$) of a linear dynamical system, which is presumed to have a Gaussian prior distribution. A linear observer views the evolving state at selected times $t>0$ in the presence of  Gaussian measurement noise. The primary goal
of this work is then to propose a new method of model reduction for linear dynamical systems in this setting that exploits natural connections between (i) the optimal posterior approximation strategy of~\cite{spantini_optimal_2015} and (ii) the system-theoretic model reduction method of balanced truncation. Our contributions are:

\begin{enumerate}
    \item We define a notion of prior \emph{compatibility} with system dynamics that allows the prior covariance matrix to be interpreted as a reachability Gramian in the inference context. 
    \item We propose two strategies for choosing a compatible prior: (i) by subjecting the system to a stochastic driving force at negative times to ``spin up'' the system state to a stationary distribution, and (ii) by modifying an incompatible prior to make it compatible. The spin-up strategy is related to ideas from the work~\cite{redmann_balanced_2018} which considers stochastic dynamical systems.    
    \item We define a \emph{noisy observability Gramian} that modifies the usual system-theoretic observability Gramian definition to account for additive noise in measurements. 
    \item We argue that the noisy observability Gramian is a natural continuous-time analogue of the Fisher information matrix, and we introduce two observation processes under which the Fisher information matrix recovers (in expectation\slash in the continuum limit)  the noisy observability Gramian up to a scaling constant.
    \item We propose a model reduction approach for our inference setting that applies the method of balanced truncation using a compatible prior covariance and the noisy observability Gramian, which inherits system-theoretic stability and error guarantees.
    \item We show that the proposed approach can recover the optimal posterior covariance approximation of~\cite{spantini_optimal_2015} in certain cases.
    \item We demonstrate numerically that the proposed approach yields near-optimal approximations when the noisy observability Gramian and Fisher information matrix define similar approximation subspaces, which occurs in our numerical experiments when measurements are closely spaced over a long observation interval. Even when measurements are limited and the approximation subspaces are not particularly close, the proposed approach yields approximations that achieve order-of-magnitude system state dimension reduction with low, albeit sub-optimal, errors. %yields
\end{enumerate} 

We note that~\cite{spantini_optimal_2015} shows that the optimal low-rank posterior update coincides with an oblique projection of the high-dimensional forward map $\bG$. Our contributions extend this result in two ways: first, we show that the low-rank update is in fact immediately associated with a \emph{low-dimensional} reduced model for the linear dynamical system. This enables efficient evolution of the reduced dynamical system for low rank inference in settings where the map $\bG$ is only available implicitly through a high-dimensional linear dynamical systems model. Second, we show that with certain observation models, this reduced model has bounded error and inherits stability.
We also note an earlier work~\cite{lieberman2013goal} developed a goal-oriented low-rank approach to classical and statistical inverse problems that was also inspired by balanced truncation but departs from the dynamical system setting: the parameter is inferred in a low-dimensional subspace defined by an `experiment observability Gramian' and a goal-oriented `prediction observability Gramian'. 
In contrast, we show that balanced truncation for dynamical systems can be used to develop reduced dynamical systems with desirable properties. 

\Cref{sec: Setting} introduces our Bayesian inference problem for the initial condition of a linear dynamical system, summarizes the balanced truncation approach to model reduction, and presents the low rank inference procedure via optimal posterior updates from \cite{spantini_optimal_2015}. \Cref{sec:low-rank-models} introduces our Gramian definitions for the inference setting, proposes the balanced truncation posterior approximation approach based on these Gramians, and proves results regarding the error, stability, and optimal posterior covariance approximation of the resulting model. These theoretical results are based on specific relationships between system theoretic Gramians and matrices appearing in inference that are discussed in detail in \Cref{sec: discussion}.  \Cref{sec: num exp} presents and discusses numerical demonstrations of the proposed method on two benchmark problems. \Cref{sec: conclusions} summarizes our discussion and considers directions for further investigation.

\section{Setting and background}\label{sec: Setting}
In this section, we introduce our notation (\Cref{ssec: notation}) and formulate the Bayesian inverse problem we consider for a high-dimensional linear dynamical system (\Cref{ssec: formulation}). We summarize two approaches to dimension reduction relevant in our setting: the system-theoretic approach to model reduction of linear time-invariant systems (\Cref{ssec: BT bg}) and the low rank inference approach of~\cite{spantini_optimal_2015} (\Cref{ssec: spantini}).

\subsection{Notation}\label{ssec: notation}
Vectors will be denoted by boldface lowercase letters, e.g., $\bv$, and matrices by boldface uppercase letters, e.g., $\bV$.  We will use $\|\bv\|$ to denote the Euclidean ($\ell^2$) norm, and $\|\bv\|^{2}_{\bM}=\bv^\top\bM\bv$ for symmetric positive (semi-)definite $\bM$ to denote the squared weighted $\bM$-(semi)norm. For square matrices $\bA$ and $\bB$ of the same size, $\bA \succ \bB$ and $\bA \succeq \bB$ mean that $\bA - \bB$ is positive definite and positive semidefinite, respectively, with analogous meanings for $\prec$ and $\preceq$.
The positive integers $\ns$, $\nin$, and $\nout$ 
denote the dimensions of state, input, and output vectors of our dynamical system. 
Random vectors (i.e., vectors of random variables) will not be denoted differently from (deterministic) vectors, but their properties will be noted when they are introduced. 
The expectation of a random vector $\bx$ is denoted as $\E[\bx]$ and the covariance of two random vectors $\bx,\by$ is denoted $\cov[\bx,\by]$.

\subsection{Problem formulation}\label{ssec: formulation}
We consider the following linear dynamical system:
\begin{align}
  \frac{\difd\bx}{\difd t} = \bA\bx,
  \label{eq: dxdt is Ax}
\end{align}
where $\bx(t)\in\R^d$ and $\bA\in\R^{d\times d}$. The initial state $\bx(0) = \bp$ is unknown. Noisy measurements $\bm_i$ of the system output are taken at times $0<t_1<t_2<\ldots<t_n$ according to the following measurement model:
\begin{align}
  \bm_i \equiv \bC \bx(t_i) + \beps_i = \bC e^{\bA t_i}\bp + \beps_i,
  \label{eq: measurements}
\end{align}
where $\bC\in\R^{\nout\times\ns}$ is the state-to-output operator and $\beps_i\sim\cN(\bzero,\gmeas)$ represents independently and identically distributed additive Gaussian noise, with positive definite covariance $\gmeas\in\R^{\nout \times\nout}$. 
In our analysis, we will consider both finite $n$ and $n = \infty$ measurements, and we will allow the observations times $\bt = (t_1, \ldots, t_n)^\top$ to be either deterministic or randomly distributed.  Random observation times are to be interpreted in the sense that they are uncertain at the beginning of the observation process, but become known as measurements are taken.  The random observation model is introduced principally to facilitate discussion of the linkage between observability Gramians and Fisher information in \Cref{sec:randgram}.

Our task is to infer the unknown initial state $\bx(0) = \bp$ from the noisy measurements in \cref{eq: measurements}. We use a Bayesian statistical approach to do so; i.e., we take as given a Gaussian prior:
\begin{align}
  \bp \sim \cN(\mathbf{0},\gprior)
    \label{eq: prior}
\end{align}
with prior covariance $ \gprior \in \R^{\ns \times \ns}$. The Gaussian prior, combined with our linear Gaussian measurement model (\cref{eq: measurements}), yields the following Gaussian likelihood conditioned on the measurement times:
\begin{align}
  \bm \mid (\bp,\bt) \sim \cN(\bG\bp, \gobs),
\end{align}
where $\bm\in\R^{n\nout}$, $\bG\in\R^{n\nout\times \ns}$, and $\gobs\in\R^{n\nout \times n\nout}$ are given as follows:
\begin{align}\label{eq:mdyn-distribution}
  \mbox{$
  \bm \equiv \begin{bmatrix} \bm_1 \\ \vdots \\ \bm_n \end{bmatrix},
  $ }
  \qquad
  \bG \equiv \begin{bmatrix}
    \bC \expe^{\bA t_1} \\ 
    \vdots \\
    \bC \expe^{\bA t_n}
  \end{bmatrix}, 
  \qquad
  \gobs \equiv \begin{bmatrix}
        \gmeas & \\
        &  \gmeas & \\
        & & \ddots & \\
        & & &  \gmeas
    \end{bmatrix}.
\end{align}

The Bayesian posterior distribution (conditioned on the measurement times) for such a prior and likelihood is again Gaussian:
\begin{align}\label{eq:posterior}
  \bp \mid (\bm, \bt) \sim \cN(\mpos,\gpos),
\end{align}
where
\begin{align}
    \mpos = \gpos \bG^\top \gobs^{-1}\bm, \qquad \gpos = \left( \bH + \gprior^{-1}\right)^{-1},
    \label{eq: posterior mean and cov}
\end{align}
where $\bH\in\R^{\ns\times \ns}$ denotes the Fisher information matrix of $\bs{m}$ with respect to the parameter $\bx_0$:
\begin{align}
    \bH \equiv \bG^\top\gobs^{-1}\bG = \sum_{i=1}^{\nmeas} \bG_i^\top \gmeas^{-1} \bG_i = \sum_{i=1}^{\nmeas} \expe^{\bA^\top t_i} \bC^\top \gmeas^{-1} \bC \expe^{\bA t_i}.
    \label{eq: fisher info sum}
\end{align}
In practical settings, the high state dimension $\ns$ can pose challenges to the computation of the conditional posterior statistics \eqref{eq: posterior mean and cov}, especially when $\bG$ is available only implicitly through a computationally expensive model that evolves the high-dimensional dynamical system~\eqref{eq: dxdt is Ax}. Our procedure addresses this challenge by providing a computationally efficient strategy for emulating the dynamical system through model reduction.

\subsection{System-theoretic model reduction via balanced truncation}\label{ssec: BT bg}
Balanced truncation~\cite{moore1981principal,mullis1976synthesis} is a system-theoretic method for projection-based model reduction that yields a cheaply evaluated
low-dimensional reduced model whose input-output map approximates that of the original high-dimensional system. \Cref{ssec: sys theory background} describes the system theory setting and \Cref{sssec: BT background} presents the balanced truncation method. 

\subsubsection{System theory background}\label{ssec: sys theory background}
The usual setting in which balanced truncation is considered involves a linear time-invariant (LTI) system given by
\begin{subequations}\label{eq: LTI dyn sys}
\begin{align}
  \frac{\difd\bx}{\difd t} &= \bA\bx + \bB \bu(t)\\
  \by &= \bC \bx
\end{align}
\end{subequations}
where $\bu(t)$ is a $\nin$-dimensional input which drives the system filtered through the input port, $\bB\in\R^{\ns \times \nin}$, and $\bC\in\R^{\nout\times \ns}$ describes the $\nout$-dimensional observer output, $\by(t)$. This system-theoretic setting differs from our inference setting in \cref{eq: dxdt is Ax} through the presence of an input $\bu(t)$ and the lack of noise in the output $\by(t)$ (compared to the measurements $\bm_i$ defined in \cref{eq: measurements}).

We assume that the system \eqref{eq: LTI dyn sys} is  stable, i.e., the eigenvalues of $\bA$ lie in the open left half-plane.
The stable system~\eqref{eq: LTI dyn sys} has infinite \emph{reachability} Gramian $\bPinf\in\R^{\ns\times \ns}$ and \emph{observability} Gramian $\bQ_{\by}\in\R^{\ns\times\ns}$, given by
\begin{align}
\bPinf &\equiv \int_0^\infty e^{\bA t} \bB \bB^\top e^{\bA^\top t} \,\difd t, &
\Qy &\equiv \int_0^\infty e^{\bA^\top t} \bC^\top \bC e^{\bA t} \,\difd t.
\label{eq: standard infinite gramians}
\end{align}
We use subscripts $\cdot_{\infty}$ and $\cdot_{\by}$ to denote the specific Gramians defined in \eqref{eq: standard infinite gramians}, contrasting with generic reachability and observability Gramians, $\bP$ and $\bQ$, which could be defined either as above or in other ways. Note in particular that $\Qy$ denotes an observability Gramian tied to the LTI output $\by$~\eqref{eq: LTI dyn sys}, whereas we will later introduce a Gramian $\Qm$ for noisy measurements $\bm$~\eqref{eq: measurements}.
The Gramians $\bPinf,\Qy$ are unique solutions to the dual Lyapunov equations
\begin{align} \label{LyapEqns}
    \bA\bPinf + \bPinf \bA^\top  = -\bB\bB^\top, &&
    \bA^\top\Qy + \Qy\bA = -\bC^\top \bC.
\end{align}
These infinite-horizon Gramians for the dynamical system \eqref{eq: LTI dyn sys} are associated with a reachability energy 
and an observability energy, 
defined as quadratic forms with respect to a state $\hat{\bx}\in\R^{d}$, respectively, as
\begin{align}\label{eq:system-energies}
    \|\hat{\bx}\|_{\bPinf^{-1}}^2 \equiv \hat{\bx}^\top\bPinf^{-1}\hat{\bx}, &&
    \|\hat{\bx}\|_{\Qy}^2 \equiv \hat{\bx}^\top \Qy\hat{\bx}.
\end{align}
The reachability energy $\|\hat{\bx}\|^2_{\bPinf^{-1}}$ is the minimum $\mathrm{L}^2$-norm $\int_{0}^\infty \|\bu(t)\|^2 \ \difd t$ of a control input, $\bu(t)$, required to steer the system from the origin ($\bx(0)=\mathbf{0}$) to $\hat{\bx}$, potentially over an infinite time horizon. Similarly, observability energy $\|\hat{\bx}\|_{\Qy}^2$ is the $\mathrm{L}^2$-norm  $\int_{0}^\infty \|\by(t)\|^2 \ \difd t$ of the (noise-free) output signal over an infinite time horizon generated by the free evolution of the system initialized at $\bx(0)=\hat{\bx}$ with null input, $\bu(t)=0$.  While the Gramian definitions in~\cref{eq: standard infinite gramians} are perhaps the most commonly used, finite-horizon and {discretized/approximate Gramian definitions} are also sometimes employed, the latter especially for large-scale problems where solving \cref{LyapEqns} is computationally intractable. For details, see, e.g., {\cite{antoulas2005approximation,BennerB2017,benner2013numerical,gugercin2004survey}}.

A system is \emph{reachable} if $\bPinf$ is positive definite (denoted as $\bPinf\succ 0$) and a system is \emph{observable} if $\Qy\succ0$. The system is \emph{minimal} if it is both reachable and observable. In all that follows, we assume that the system is minimal.

\subsubsection{Balanced truncation}\label{sssec: BT background}
In projection-based model reduction, a reduced model is obtained by projecting the system operators onto low-dimensional subspaces. In balanced truncation, these subspaces are determined by the dominant eigenvectors of the matrix pencil $(\bQ,\bP^{-1})$, i.e., the $\bv_i$ associated with the dominant eigenvalues, $\delta_i^2$, satisfying
\begin{align}
         \bQ\, \bv_i=\delta_i^2 \ \bP^{-1}\, \bv_i,
         \label{eq: BT GEV problem}
\end{align}
where $\bP$ is a reachability Gramian and $\bQ$ is an observability Gramian: the standard balanced truncation approach takes $\bP=\bPinf$ and $\bQ=\Qy$, but we will in \Cref{sec:low-rank-models} introduce approaches that use different Gramians. The scalars $\left\{\delta_i\right\}_{i=1}^d$ with $\delta_1 > \delta_2> \cdots >\delta_\ns$ are called the \textit{Hankel singular values} and are the nonzero singular values of the underlying Hankel operator. We assume that the Hankel singular values are distinct to simplify the presentation. 
The dominant eigenvectors, $\bv_i$, maximize the following Rayleigh quotient:
\begin{align}
    \frac{\bv^\top \bQ \bv}{\bv^\top \bP^{-1} \bv} = \left(\frac{\|\bv\|_{\bQ}}{\|\bv\|_{\bP^{-1}}}\right)^2.
    \label{eq: rayleigh quotiet}
\end{align}
Thus, directions retained by balanced truncation will be simultaneously easy to reach (low reachability energy) and easy to observe (high observability energy), or will realize some trade-off between these two desirable properties. 

The balanced truncation procedure has the following explicit algebraic formulation. Suppose $\bP = \bR\bR^\top$ and $\bQ = \bL\bL^\top$ (e.g., through a Cholesky factorization), and let $\bU\bDelta\bZ^\top= \bL^\top\bR$ be the singular value decomposition of $\bL^\top\bR$. Then $\bDelta$ contains the Hankel singular values,
$  \bDelta = \diag\left(\delta_1, \ldots, \delta_d\right). $
Define $\bT = \bR\bZ\bDelta^{-\frac12} \in \R^{\ns\times\ns}$ and note that $\bT$ is \emph{invertible}
with  $\bT^{-1} = \bDelta^{-\frac12}\bU^\top\bL^\top$. 
The ordered generalized eigenvectors, $\bv_i$, of $(\bQ, \bP^{-1})$ are the columns of $\bV=\bR\bZ$, and through a renormalization by $\bDelta^{-\frac12}$ become the columns of $\bT$.  

 $\bT$ is a \textit{balancing} transformation, meaning that in the transformed coordinates $\widetilde\bx=\bT^{-1}\bx$ the reachability and observability Gramians are equal diagonal matrices.  The transformed LTI system is given by,
\begin{subequations}\label{eq: LTI balanced}
\begin{align}
  \frac{\difd\widetilde{\bx}}{\difd t} &= \widetilde{\bA}\widetilde{\bx} + \widetilde{\bB} \bu(t)\\
  \widetilde{\by} &= \widetilde{\bC} \widetilde{\bx},
\end{align}
\end{subequations}
with
\begin{align*}
  \bT \widetilde{\bx} &= \bx, & \widetilde{\bA} &= \bT^{-1} \bA \bT, & \widetilde{\bB} &= \bT^{-1} \bB, & \widetilde{\bC} &= \bC \bT,
\end{align*}
and the system \eqref{eq: LTI balanced} is called (principal-axis) \emph{balanced} because its reachability and observability Gramians $\widetilde{\bP}$ and $\widetilde{\bQ}$ are equal and diagonal, satisfying
\begin{align}\label{eq:balanced-gramians}
  \bT^{-1} \bP \bT^{-\top} = \widetilde{\bP} = \bSigma = \widetilde{\bQ} = \bT^\top \bQ \bT.
\end{align}
Order-$r$ balanced truncation obtains a reduced dynamical system by transforming to the balancing coordinates $\widetilde\bx$ and truncating to the leading $r$ components of this balanced state. 
Denote by $\bU_r,\bZ_r\in\R^{d\times r}$ the leading $r$ columns of $\bU,\bZ$,  respectively, and let $\bDelta_r$ denote the upper-left $r\times r$ block of $\bDelta$. Then, let $\bT_r$ denote the leading $r$ columns of $\bT$ and let $\bS_r^\top$ denote the leading $r$ rows of $\bT^{-1}$, i.e.,
\begin{align}
    \bT_r = \bR\bZ_r\bDelta_r^{-\frac12} \in \R^{d\times r}, \qquad \bS_r^\top  =\bDelta_r^{-\frac12}\bU_r^\top\bL^\top\in\R^{r\times d}.
    \label{eq: balancing bases}
\end{align}
{Note that $\bS_r^\top \bT_r = \bI_r$.}
Then, the order-$r$ balanced truncation approximation of~\eqref{eq: LTI dyn sys} is a Petrov-Galerkin approximation given by restricting the state to lie in $\mathsf{range}({\bT_r})$, and by enforcing~\eqref{eq: LTI dyn sys} on $\mathsf{range}({\bS_r})^{\perp}$. This results in the reduced-order approximation of \cref{eq: LTI dyn sys},
\begin{subequations}\label{eq: bal red dynamics}
  \begin{align}
    \frac{\difd \bx_r}{\difd t} &= \bA_r\bx_r + \bB_r \bu(t), \\
    \by_r &= \bC_r \bx,
  \end{align}
  \label{eq: BT reduced dyn sys}
  \end{subequations}
where $\bA_r = \bS_r^\top\bA\bT_r\in\R^{r\times r}$, $\bB_r = \bS_r^\top\bB\in\R^{r\times\nin}$, $\bC_r = \bC\bT_r\in\R^{\nout\times r}$, and $\bT_r\bx_r\approx\bx$. 

Balanced truncation provides guaranteed stability and \textit{a priori} $L_2$-error bounds on the approximation error expressed
in terms of the Hankel singular values in \eqref{eq: BT GEV problem}:
\begin{lemma}[Thm. 7.9 of \cite{antoulas2005approximation}]\label{lemma: continuous time BT}
  Consider the system \eqref{eq: LTI dyn sys} with output $\by$ and reachability and observability Gramians $\bPinf$ and $\Qy$ defined in~\eqref{eq: standard infinite gramians}. If the original system~\eqref{eq: LTI dyn sys} is {stable}, then the reduced model~\eqref{eq: BT reduced dyn sys} is both stable and balanced: the infinite-horizon reachability and observability Gramians of the reduced state, $\bx_r$, are diagonal and equal, and given by
  \begin{align} \label{eq:redgrams}
     {\bP}_r \equiv  \int_0^\infty  e^{\bA_rt} \bB_r \bB_r^\top e^{\bA_r^\top t}  dt =   \int_0^\infty  e^{\bA_r^\top t} \bC_r^\top \bC_r e^{\bA_r t}  dt \equiv {\bQ}_{\by_r} = 
     \diag\left(\delta_1, \ldots, \delta_r \right)  \in\R^{r\times r}.
  \end{align}
  Furthermore, if $\bx(0)=0$, then the output $\by_r$ of \cref{eq: BT reduced dyn sys} commits an error bounded by twice the sum of the truncated Hankel singular values:
  \begin{align}\label{eq:hsv-bt-bound}
    \left\| \by - \by_r \right\|_{L^2(\R)} \leq \left( 2  \sum_{j=r+1}^d \delta_j \right) \left\| \bu \right\|_{L^2(\R)}. 
  \end{align}
\end{lemma}

\subsection{Low rank inference via optimal posterior covariance approximation}\label{ssec: spantini}
We now summarize the main results of~\cite{spantini_optimal_2015} for dimension reduction of the linear Gaussian Bayesian inverse problem by optimal approximation of the posterior covariance and mean. This approach is applicable even when $\bG$ does not entail the evolution of a dynamical system, and thus complements the model reduction approach of \Cref{ssec: BT bg}.

Let $(\tau_i^2,\bw_i)$ denote the ordered generalized eigenpairs of the pencil $(\bH,\gprior^{-1})$, i.e.,
\begin{align}
    \bH \bw_i = \tau_i^2\gprior^{-1}\bw_i,
\end{align}
with $\bw_i^\top \gprior^{-1}\bw_j=\delta_{i,j}$, where $\delta_{i,j}$ denotes the Kronecker delta, and $\tau_1 \geq \tau_2 \geq \cdots \geq \tau_{\ns}$.
Then, note that the posterior covariance, $\gpos$, satisfies
\begin{align}
    \gpos = (\gprior^{-1} + \bH)^{-1} = \gprior - \sum_{i=1}^d \frac{\tau_i^2}{1+\tau_i^2}\bw_i\bw_i^\top.\label{eq: full posterior eig sum}
\end{align}
When there are only a few non-zero eigenvalues, $\tau_i^2$ (e.g., when $\nobs\ll \ns$), or when $\tau_i^2$ decays rapidly with $i$, the sum in~\eqref{eq: full posterior eig sum} can be well-approximated by a low-rank symmetric positive semidefinite matrix. 
The work in~\cite{spantini_optimal_2015} thus considers the optimal  approximation of $\gpos$ within the following class of rank-$r$ negative semidefinite updates to $\gprior$:
\begin{equation}
\label{approxClass:Mr}
  \mathcal{M}_r \coloneqq \left\{\mathbf{M}\in\mathbb{R}^{\ns\times \ns}\quad \big|\quad \mathbf{M}=\gprior- \bK \bK^\top \succ 0 \;\mbox{ with }\; \mathsf{rank}(\bK) \leq r \right\}.
\end{equation}
The distance between the true posterior covariance and an approximation within $\mathcal{M}_r$ is measured by the F\"orstner metric \cite{forstner_metric_2003} between symmetric positive semidefinite matrices $\bA,\bB\in\R^{d\times d}$:
\begin{align}
  d_{\mathcal{F}}(\bA,\bB) = \sum_{i=1}^d \ln^2(\sigma_i),
\end{align}
where $\sigma_i$ are the eigenvalues of the pencil $(\bA,\bB)$. 
The F\"orstner distance is invariant to similarity transformation, and $d_{\mathcal{F}}(\bA,\bB) = d_{\mathcal{F}}(\bA^{-1},\bB^{-1})$.
The F\"orstner distance is widely used in statistics, being the natural Riemannian metric on the manifold of positive-definite (e.g., full-rank covariance) matrices \cite[Section 6.1]{bhatia_positive_2006} that can be generalized with appropriate modifications to rank-deficient manifolds of semi-definite matrices \cite{bonnabel_riemannian_2010}.

The F\"orstner-optimal approximation $\hgpos$ to $\gpos$ within the class of rank-$r$ updates $\mathcal{M}_r$,   
\begin{align}\label{eq:post-optimization}
  \hgpos=\argmin_{\bM \in \mathcal{M}_r} d_{\mathcal{F}}\left(\gpos, \bM\right) 
\end{align}
 is then given by~\cite{spantini_optimal_2015}:
\begin{align} \label{truncK}
\hgpos = \gprior - \bK_*\bK_*^\top \quad \mbox{where}\quad
    \bK_*\bK_*^\top = \sum_{i=1}^r \frac{\tau_i^2}{1+\tau_i^2}\bw_i\bw_i^\top,
\end{align}
with optimal F\"orstner distance 
\begin{align}\label{eq:forstner-minimizer}
  d_{\mathcal{F}} \left( \gpos, \hgpos \right) = \min_{\bM \in \mathcal{M}_r} d_{\mathcal{F}}\left(\gpos, \bM\right) = \sum_{i=r+1}^d \ln^2\left(\frac{1}{1 + \tau_i^2}\right).
\end{align}

An alternative analysis in~\cite{spantini_optimal_2015} allows one to associate the optimal approximation~\eqref{truncK} with the posterior covariance of a projected forward map.
Note that the pairs $(\frac 1{1+\tau_i^2}, \tilde\bw_i)$ with $\tilde\bw_i = \gprior^{-1}\bw_i$ are the generalized eigenpairs of the pencil $(\gpos,\gprior)$. Let $\boldsymbol{\Pi}_r = \sum_{i=1}^r \tilde\bw_i \bw_i^\top$ denote the oblique spectral projector for $\gprior^{-1} \gpos$ onto $\mathsf{span}\{\tilde\bw_1,\tilde\bw_2,\ldots,\tilde\bw_r\}$. Then, let 
\begin{align}\label{eq:projected-model-and-FI}
  \widehat\bG &=\bG\boldsymbol{\Pi}_r^\top, & 
  \widehat\bH &= \widehat\bG^\top\gobs^{-1}\widehat\bG, & 
\end{align}
denote the projected forward map and projected Fisher information, respectively.
Then, the optimal posterior covariance approximation~\eqref{truncK} can be equivalently expressed as
\begin{align}\label{eq: spantini cov H}
    \hgpos = (\gprior^{-1} + \widehat{\bH})^{-1}.
\end{align}

The work in~\cite{spantini_optimal_2015} further considers optimal approximation to the posterior mean $\mpos$ in two different classes of mean approximations: a class of low-rank (LR) mean approximations given by
\begin{align}\label{eq: low rank mean class}
    \mathcal{V}_r^{\text{LR}} \coloneqq \{ \bv = \bN \bm \mid \bN\in\R^{d\times n\nout}, \mathsf{rank}(\bN) \leq r\}
\end{align}
and a class of mean approximations given by low-rank update (LRU) approximations to the posterior covariance:
\begin{align}\label{eq: low rank update class}
    \mathcal{V}_r^{\text{LRU}} \coloneqq \{ \bv = (\gprior - \bN)\bG^\top\gobs^{-1}\bm \mid \bN\in\R^{d\times d},\, \mathsf{rank}(\bN)\leq r\}.
\end{align}
Within $\mathcal{V}_r^{\rm LR}$, the class of low-rank mean approximations~\eqref{eq: low rank mean class}, the following mean approximation is optimal in the sense that it minimizes 
the $\|\cdot\|_{\gpos^{-1}}$-norm error  averaged over the joint distribution of the initial condition (with marginal given by the prior) and the data  (e.g., Bayes risk):
\begin{align}\label{eq: low rank mean}
    \hmpos^{\text{LR}} = \sum_{i=1}^r\frac{\tau_i}{1+\tau_i^2} \bw_i\tilde\bw_i^\top = \hgpos\widehat{\bG}^\top \gobs^{-1}\bm.
\end{align}
Within $\mathcal{V}_r^{\rm LRU}$, the class of low-rank update mean approximations~\eqref{eq: low rank update class}, the optimal mean approximation in the $\|\cdot\|_{\gpos^{-1}}$-norm Bayes risk is given by~\cite{spantini_optimal_2015}:
\begin{align}\label{eq: low rank update mean}
    \hmpos^{\text{LRU}} = \hgpos\bG^\top \gobs^{-1}\bm.
\end{align}
Note that in a dynamical systems setting, computation of the posterior approximations $\hgpos$, $\hmpos^{\rm LR}$, and $\hmpos^{\rm LRU}$ requires evolving the full high-dimensional dynamics of $\bG$ and then projecting the result to obtain $\widehat\bG$. In the next section, we propose a model reduction approach that yields a low-dimensional dynamical systems model by combining elements from this posterior covariance approximation approach and from balanced truncation.

\section{Balanced truncation for Bayesian inference}\label{sec:low-rank-models}
We now explore the application of the balanced truncation procedure of \Cref{ssec: BT bg} to obtain a reduced model for the linear dynamical system \eqref{eq: dxdt is Ax} in an effort to reduce computational cost in our Bayesian inverse problem.  To do this, we first negotiate two technical hurdles which frustrate a direct connection:  the dynamical system \eqref{eq: dxdt is Ax} lacks an input and thus a notion of reachability, and the standard infinite-horizon observability Gramian $\Qy$ accounts for neither the additive measurement noise nor the specific observation times of the measurement model~\eqref{eq: measurements}.  To reconcile these discrepancies with the inference problem, we introduce suitably modified Gramian definitions in \Cref{ssec: inference Gramians}.  In \Cref{ssec: method}, we then present two variations of our proposed balanced truncation approach to Bayesian inference.  In~ \Cref{ssec: properties}, we establish conditions under which one variation of our proposed method will approximately recover the optimal approximation qualities of the low rank inference procedure established in~\cite{spantini_optimal_2015}, with a faster computation of the forward maps via a stable and balanced reduced approximation of the dynamical system.  The second variation we present may appear more natural, though it does not inherit the stability guarantees of a proper balanced truncation.  Both variations will be explored in numerical examples in \Cref{sec: num exp}.

\subsection{Inference Gramians}\label{ssec: inference Gramians}
We begin by introducing a definition of \emph{prior-compatibility} which allows the prior covariance, $\gprior$, to be interpreted as a reachability Gramian of a forced system related to~\eqref{eq: dxdt is Ax}.
\begin{definition}
  \label{def: compatible prior}
 The prior covariance, $\gprior$, will be said to be \emph{prior-compatible} with state dynamics induced by $\bA$ if  $\bA\gprior + \gprior\bA^\top$ is negative semidefinite.
\end{definition}
If a prior covariance $\gprior$ satisfies \Cref{def: compatible prior}, then there exists some $\bB\in\R^{\ns\times\nin}$ with $\nin = \mathsf{rank}(\bA\gprior +\gprior\bA^\top)$ such that 
the Lyapunov equation~\eqref{LyapEqns} $\bA\gprior + \gprior \bA^\top = -\bB\bB^\top$ is satisfied. Thus, the prior-compatibility with state dynamics means that the prior covariance matrix $\gprior$ is the reachability Gramian of an LTI system with dynamics specified by $\bA$ and an input port matrix $\bB$. 
It is not generally necessary to identify a particular port matrix, $\bB$, for \Cref{def: compatible prior} to hold, although in some applications $\bB$ may be naturally provided (for example if the end goal of the inference problem is to determine a control input).
Discussions of how a compatible prior might be either naturally defined or obtained by modifying a non-compatible prior are provided in~\Cref{ssec:prior compatibility}.

We make use of a modified observability Gramian that is motivated by the noisy measurement model for $\bm$. Recall that the standard observability Gramian $\Qy$ defines an energy that is the squared ${\rm L}^2$-norm of the output signal $\by(t)$. In our inference context, where the measurement at time $t_i$, $\bm_i$, is an output polluted by measurement noise $\beps_i\sim \mathcal{N}(0,\gmeas)$, we
quantify the discernability of the signal from the noise by the Mahalanobis distance
from the equilibrium state $ \bzero$ to the conditional distribution $\cN(\bC \expe^{\bA t_i} \bp, \gmeas)$ of  $\bm_i|(\bp,t_i)$:
\begin{align}\label{eq: mahalanobis distance}
  D\left(\bzero, \cN(\bC \expe^{\bA t_i} \bp,\gmeas)\right) \coloneqq \left\| \bC \expe^{\bA t_i} \bp \right\|_{\gmeas^{-1}}.
\end{align} 
One way to understand why this is a natural way to measure the distance between a distribution (the signal with observation noise) and a point (the stable equilibrium) is to note that for the case of a Gaussian distribution, neighborhoods of the mean of the distribution defined by the Mahalanobis distance contain the most probability relative to any other set with the same Lebesgue measure (Appendix~\ref{sec:chrishomework}).
This motivates our definition of a modified observability Gramian for noisy measurements.
\begin{definition}
  \label{def: noisy obs Gramian}
    The noisy observability Gramian $\Qm\in\R^{\ns \times \ns}$ is given by
    \begin{align}
    \Qm \equiv \int_0^\infty \expe^{\bA^\top t}\bC^\top \gmeas^{-1}\bC\expe^{\bA t}\,\difd t.
\end{align}
\end{definition}
Thus, $\Qm$ defines an energy based on integrating over all positive time the squared Mahalanobis distance between the equilibrium and the noise-polluted outputs. 
Note that the Fisher information matrix $\bH$~\eqref{eq: fisher info sum} can be viewed as defining an analogous energy based on summing over discrete measurement times. This analogy motivates an     approximation of $\bH$ in terms of $\Qm$ discussed in subsequent sections.

\subsection{Method}\label{ssec: method}
We propose a model reduction approach for the inference problem defined in \Cref{ssec: formulation} that applies the balanced truncation procedure of \Cref{ssec: BT bg} by taking first a compatible prior covariance, $\gprior$, and then using as inference Gramians $\bP=\gprior$ and $\bQ=\Qm$. This leads to the following reduced-order model for the linear dynamics~\eqref{eq: dxdt is Ax}:
\begin{align} \label{eq: bal for inf red dyn}
    \frac{\difd\bx_r}{\difd t} = \bA_r\bx_r, \qquad \bx_r(0) = \bS_r^\top\bp, 
\end{align}
where $\bA_r = \bS_r^\top\bA\bT_r$, and the following reduced-order approximation to the measurement model~\eqref{eq: measurements}: 
\begin{align}\label{eq: bal for inf meas model}
    \bm_i \approx \bC_r\bx_r(t_i) +\beps_i,
\end{align}
where $\bC_r = \bC\bT_r$, and $\beps_i\sim\mathcal{N}(0,\gmeas)$ as before. 
The reduced dynamics \eqref{eq: bal for inf red dyn} and measurement model \eqref{eq: bal for inf meas model} define a reduced forward map from the reduced state to the noise-free output:
\begin{align}
    \bG_{BT} \equiv \begin{bmatrix}
        \bC_r \expe^{\bA_rt_1}\bS_r^\top \\ \vdots \\ \bC_r \expe^{\bA_rt_n} \bS_r^\top
    \end{bmatrix} = \begin{bmatrix}
        \bC_r \expe^{\bA_rt_1}\\ \vdots \\ \bC_r \expe^{\bA_rt_n} 
    \end{bmatrix}\bS_r^\top 
    \in\R^{n\nout\times \ns}.
\end{align}
We then define a corresponding reduced-order approximate Fisher information matrix, 
\begin{align}\label{eq: BT fisher info}
  \bH_{\mathrm{BT}} \equiv \bG_{\rm BT}^\top \gobs^{-1} \bG_{\rm BT} = \bS_r \left(\sum_{i=1}^{\nmeas} e^{\bA_r^\top t_i} \bC_r^\top \gmeas^{-1} \bC_r e^{\bA_r t_i}\right) \bS_r^\top,
\end{align}
and associated posterior mean and covariance approximations:
\begin{align}\label{eq:bt-mean-and-cov}
  \mposBT = \gposBT \bG_{\mathrm{BT}}^\top \gobs^{-1}\bm, \qquad \gposBT = \left( \bH_{\mathrm{BT}} + \gprior^{-1}\right)^{-1}.
\end{align}
As previously discussed at the end of \Cref{ssec: inference Gramians}, the noisy observability Gramian $\Qm$ can be viewed as a continuous-time analogue of $\bH$. 
This suggests a variation of the proposed method that uses the inference Gramians $\bP=\gprior$ and $\bQ=\bH$ to define the balanced truncation. We label the two variations of our proposed method as follows:
\begin{enumerate}
    \item \textbf{BT-Q}, which uses the inference Gramians $(\bP,\bQ)=(\gprior,\Qm)$ in order to define the posterior approximations in~\eqref{eq:bt-mean-and-cov}, and
    \item \textbf{BT-H}, which uses the inference Gramians $(\bP,\bQ)=(\gprior,\bH)$  in order to define the posterior approximations in~\eqref{eq:bt-mean-and-cov}.
\end{enumerate}
While the BT-H approach may seem to be a more natural choice for our inference setting, we will show in \Cref{ssec: properties} that the BT-Q approach inherits desirable properties from standard system-theoretic balanced truncation that the BT-H approach generally does not enjoy.
\newline

\subsection{Properties}
\label{ssec: properties}
We note that while the principal goal of the balanced truncation approach in the present work is to provide a reduced model approximation to posterior quantities, the proposed BT-Q approach inherits guarantees from the system-theoretic LTI setting, as described in our first result.
\begin{lemma}\label{lem: BT Goodness}
    Suppose $\gprior$ is prior-compatible with the state dynamics~\eqref{eq: dxdt is Ax} as defined in \Cref{def: compatible prior}. Then, there exists a $\bB\in\R^{\ns\times\nin}$ so that $\gprior$ and $\Qm$ are the infinite reachability and observability Gramians of the following LTI system:
    \begin{align}
        \dv{\bx}t = \bA\bx + \bB \bu(t), \qquad \bar\by(t) = \bar\bC\bx(t),\label{eq: modified LTI system}
    \end{align}
    where $\bar\bC = \gobs^{-\frac12}\bC$.
    
    Additionally, the balanced reduced LTI system~\eqref{eq: BT reduced dyn sys}
based on the inference Gramians $\gprior$ and $\Qm$ is balanced, stable, and satisfies the following Hankel singular value tail bound: if $\bx(0)=0$, then 
    \begin{align}
        \|\bar\bC \bx(t) - \bar\bC_r\bx_r(t)\|_{L^2(\R)} \leq \left(2\sum_{j=r+1}^d \delta_j\right) \|\bu\|_{L^2(\R)},
    \end{align}
    where $\delta_j$ are the Hankel singular values associated with the inference Gramians $\gprior$ and $\Qm$.
\end{lemma}
\begin{proof}
    The first part of \Cref{lem: BT Goodness} follows immediately from the compatibility condition in \Cref{def: compatible prior} and from replacing $\bC$ with $\bar\bC$ in \cref{eq: standard infinite gramians}. The second part then follows from direct application of \Cref{lemma: continuous time BT} to \eqref{eq: modified LTI system}. \hfill$\Box$
\end{proof}

\begin{corollary}\label{cor: inference stability}
    The proposed balanced truncation reduced model~\eqref{eq: bal for inf red dyn} for the inference dynamical system~\eqref{eq: dxdt is Ax} is stable.
\end{corollary}

\Cref{cor: inference stability} has practical significance for the proposed BT-Q approach of \Cref{ssec: method}, because computing the approximate posterior quantities requires evolving the reduced dynamics given by~\eqref{eq: dxdt is Ax}: \Cref{cor: inference stability} guarantees that these reduced dynamics preserve the stability of the full dynamics. On the other hand, while the error and balancing guarantees of \Cref{lem: BT Goodness} for LTI systems are not directly applicable to the inference setting we consider here, they may serve as the basis for future extension to inference problems with control inputs.

We now turn our attention to the quality of the proposed posterior covariance approximation~\cref{eq:bt-mean-and-cov}: again, we focus on the BT-Q approach. The BT-Q posterior covariance approximation must tautologically be suboptimal relative to the optimal approximation scheme of~\cite{spantini_optimal_2015} discussed in \Cref{ssec: spantini}. A natural question is: when does covariance approximation approach optimality? In comparing the covariance approximations~\eqref{eq:bt-mean-and-cov} and~\eqref{eq: spantini cov H}, it is clear that as $r\to d$, both $\bH_{\BT}\to\bH$ and $\widehat{\bH}\to\bH$; i.e., both the balanced truncation Fisher information approximation and the projected Fisher information of~\cite{spantini_optimal_2015} converge to the true Fisher information and thus the two posterior covariance approximations both converge to the full posterior covariance. However, since our goal is a low rank approximation with $r\ll d$, we would like some sense of when we can expect $\bH_{\rm BT}$ to be close to $\widehat{\bH}$ for general $r$. 

For $\bH_{\rm BT}$ resulting from the BT-Q approach, approximate agreement with $\widehat\bH$ requires that the dominant $r$ generalized eigenvectors of the pencil $(\bH, \gprior^{-1})$ be close to the dominant $r$ generalized eigenvectors of the pencil $ (\Qm, \gprior^{-1})$.  This would follow when the Fisher information $\bH$~\eqref{eq: fisher info sum} is approximately proportional to the noisy observability Gramian $ \bQ_m$ in~\Cref{def: noisy obs Gramian}, that is $ \Qm \approx h \bH$ for a scaling factor $h$.  This would seem to apply to situations in which the observations are of sufficient frequency and duration for the discrete sum in $\bH$ to provide, up to normalization, a good approximation to the integral in $\bQ_m$.  In \Cref{ssec: observability discussion}, we formalize this idea and identify the scaling factor $h$ as the (average) inter-observation interval and establish two limits of observation protocols under which we can obtain $ \Qm $ as a rescaled limit of $ \bH$.

Even when the spectral projections based on $ \bH $ and $ \Qm $ are close, there remains a difference between $ \widehat{\bH}$ and $ \bH_{\rm BT} $ because the former involves the evolution of the full dynamics represented by matrix $\bA$, whereas the proposed balanced truncation approximations require only the evolution of a computationally inexpensive reduced system $ \bA_r$.  
Remarkably, we can formulate a relation between the balanced truncation approximation to the Fisher information and the optimal low rank inference approximation $ \widehat{\bH}$ without a further assumption about the size of $r$.

To express this connection precisely, we  consider the following noisy observability Gramian for the balanced truncation approximation:
\begin{align}\label{eq: QmBT def}
    \Qm^{\BT} \equiv \bS_r \left(\int_0^\infty \expe^{\bA_r^\top t}\bC_r^\top\gmeas^{-1}\bC_r\expe^{\bA_r t} \,\difd{t} \right)\bS_r^\top \in\R^{\ns\times\ns}.
\end{align}
The quantity $\Qm^{\BT}$ can be viewed as a continuous-time analogue of $\bH_{\BT}$: it lifts the reduced Gramian for the reduced state $\bx_r\in\R^r$ to the full state $\bx\in\R^{\ns}$ via the map $\bS_r$. 
 Thus if the sum in $\bH$ well-approximates the integral in $\Qm$, i.e., $h\bH\approx \Qm$, we would also expect $h\bH_\BT\approx \Qm^{\BT}$ because $\bH_\BT$ and $\Qm^\BT$ have the same sum-to-integral relation as $\bH$ and $\Qm$, differing only in the use of reduced operators $\bC_r$ and $\bA_r$ vs. full operators $\bC$ and $\bA$.

\begin{proposition}\label{prop: QmBT is hatH}
    Let $(\tau_i^2,\bw_i)$ and $(\delta_i^2,\bv_i)$, respectively, denote the ordered generalized eigenpairs of the pencil $(\bH,\gprior^{-1})$ and $(\Qm,\gprior^{-1})$, respectively, with normalization $\bw_i^\top\gprior^{-1}\bw_i=1$ and $\bv_i^\top\gprior^{-1}\bv_i=1$. Then, if the $r$ dominant eigenpairs match up to a scaling constant $h > 0$, i.e., if
    \begin{align}\label{eq: eigenspace assumption}
        \delta_i^2 = h\,\tau_i^2, \quad \bw_i = \bv_i, \quad \forall i\leq r,
    \end{align}
    then $\Qm^\BT =h\,\widehat{\bH}$.
\end{proposition}

\begin{proof}
    The integral in \cref{eq: QmBT def} is the noisy observability Gramian of the balanced truncation reduced system and thus, by \Cref{lem: BT Goodness} it is balanced, i.e.,
    \begin{align}
        \int_0^\infty \expe^{\bA_r^\top t}\bC_r^\top\gmeas^{-1}\bC_r\expe^{\bA_r t} \,\difd{t} = \diag\{\delta_1,\delta_2,\ldots,\delta_r\}\equiv\bDelta_r.
    \end{align}
    Define $\bV_r=[\bv_1,\bv_2,\ldots,\bv_r]$ and note from its definition in \eqref{eq: balancing bases} that $\bS_r^\top=\bDelta_r^{\frac12}\bV_r^\top\gprior^{-1}$. Thus, \cref{eq: QmBT def} is equivalent to
    \begin{align*}
        \Qm^{\BT} =\gprior^{-1}\bV_r \bDelta_r^{2}\bV_r^\top\gprior^{-1}.
    \end{align*}
    Similarly defining $\bW_r = [\bw_1,\bw_2,\ldots,\bw_r]$, we have from \eqref{eq:projected-model-and-FI},  
    \begin{align*}
        \widehat{\bH} = \gprior^{-1}\bW_r \bSigma_r^{2}\bW_r^\top\gprior^{-1} 
    \end{align*}
    with $\bSigma_r=\diag\{\tau_1,\tau_2,\ldots,\tau_r\}$.
    Under the assumption~\eqref{eq: eigenspace assumption}, this gives us $\Qm^{\BT}= h\widehat{\bH}$.\hfill$\Box$
\end{proof}

\Cref{prop: QmBT is hatH} provides us with a bridge between $\bH_{\BT}$ and $\widehat{\bH}$ for any $r$, including $r\ll d$.   It suggests the following argument: as noted above, when $\Qm \approx h \bH$,  generally $ \Qm^{\BT} \approx h \bH_{\rm BT}$ will also be true, 
and \Cref{prop: QmBT is hatH} would imply $ \Qm^{\BT} \approx h \widehat{\bH} $, so that $ \bH_{\rm BT} \approx \widehat{\bH}$.  That is, when $ \Qm$ and $ \bH $ are approximately proportional,  the BT-Q posterior covariance approximation converges to the optimal one without further assumption on $r$.  Rigorizing such an argument would require a formulation of \Cref{prop: QmBT is hatH} with approximate rather than precise equalities of the dominant eigenpairs.  Rather than pursuing this cumbersome exercise, we instead support this argument by numerical experiments in~\Cref{sec: num exp}.  We confirm that in observation scenarios where the Fisher information matrix $ \bH$ and noisy observability Gramian $ \Qm$ are approximately proportional, the BT-Q approach leads to near-optimal posterior covariance approximations.

The preceding discussion provides intuition for what types of measurement models will lead to posterior covariance approximations that are close to the optimal low-rank approach of~\cite{spantini_optimal_2015}. Our analysis and discussion do not address the quality of the posterior mean approximation in~\eqref{eq:bt-mean-and-cov}: while the balanced truncation mean approximation does belong to the class of low rank mean approximations~\eqref{eq: low rank mean class} considered in~\cite{spantini_optimal_2015}, the balanced truncation mean approximation differs in the use of the balanced truncation adjoint approximation vs.\ the projected full adjoint used in the optimal low-rank mean~\eqref{eq: low rank mean} of~\cite{spantini_optimal_2015}. Numerical experiments in \Cref{sec: num exp} compare empirically the balanced truncation means with the optimal means of~\cite{spantini_optimal_2015}. We emphasize that in the dynamical system setting, both optimal approximations of~\cite{spantini_optimal_2015} require the evolution of the full dynamics, whereas the proposed balanced truncation approximation approach requires only the evolution of a computationally inexpensive reduced system. 

Until now we have focused on properties of the BT-Q approach, and the result of \Cref{prop: QmBT is hatH} relies on the assumption~\eqref{eq: eigenspace assumption} that the dominant eigenspaces of $\Qm$ and $\bH$ match, and that the associated eigenvalues match up to a constant scale factor. 
Comparable guarantees for the BT-H approach, which involves replacing $\Qm$ with $\bH$ in the balanced truncation approximation,
appear unlikely for general choices of Fisher information matrix, since they cannot generally be interpreted as infinite-horizon observability Gramians for any minimal LTI system, and so~\Cref{lemma: continuous time BT} does not directly apply. However, when the dominant generalized eigenvectors of the pencils $(\bH,\gprior^{-1})$ and $(\Qm,\gprior^{-1})$ coincide, then the BT-H and BT-Q reduced models will be equivalent, as stated in our next result (note that no conditions are imposed on the generalized eigen\emph{values} here).

\begin{proposition}\label{prop:Spantini=BT}
    Let $(\tau_i^2,\bw_i)$ and $(\delta_i^2,\bv_i)$, respectively, denote the ordered generalized eigenpairs of the pencil $(\bH,\gprior^{-1})$ and $(\Qm,\gprior^{-1})$, respectively, with normalization $\bw_i^\top\gprior^{-1}\bw_i=1$ and $\bv_i^\top\gprior^{-1}\bv_i=1$. Then, if the $r$ dominant eigenvectors match, i.e., if $\bw_i = \bv_i$, for all $i\leq r$,
    then the reduced model, $\bA_r = \bS_r^\top\bA\bT_r$, $\bB_r = \bS_r^\top\bB$, $\bC_r = \bC\bT_r$, produced using the state space projection $\bT_r\bS_r^\top$ defined by the pencil $(\bH,\gprior^{-1})$, replicates (up to a diagonal scaling of basis) the reduced model produced by standard balanced truncation using the pencil $(\Qm,\gprior^{-1})$.
\end{proposition}
\begin{proof}
Suppose the truncation order is $r$, and let the columns of $\mathbf{V}_r=[\bv_1,\bv_2,...,\bv_r]$ collect the dominant generalized eigenvectors of the pencil $(\Qm,\gprior^{-1})$ which are presumed to be shared with the pencil $(\bH,\gprior^{-1})$. Taking $\bDelta_r=\mathsf{diag}\{\delta_1,\delta_2,...,\delta_r\}$ and $\Sigma_r=\mathsf{diag}\{\tau_1,\tau_2,...,\tau_r\}$, we may write $\Qm\mathbf{V}_r=\gprior^{-1}\mathbf{V}_r\bDelta_r^2$ and $\bH\mathbf{V}_r=\gprior^{-1}\mathbf{V}_r\Sigma_r^2$.  Defining $\gprior = \bR\bR^\top$, $\Qm = \bL\bL^\top$, and  $\bH = \widetilde{\bL}\widetilde{\bL}^\top$,  we determine singular value decompositions $\bU\bDelta\bZ^\top= \bL^\top\bR$ and $\widetilde{\bU}\bSigma\widetilde{\bZ}^\top= \widetilde{\bL}^\top\bR$, with associated truncated matrices, $\bU_r$, $\bDelta_r$, $\bZ_r$, $\widetilde{\bU}_r$, $\bSigma_r$, and $\widetilde{\bZ}_r$.  The truncated balancing projection determined by the standard Gramians $\Qm$ and $\gprior$ are given by $\mathbf{T}_r=\bR\bZ_r\bDelta_r^{-\frac12}$ and $\mathbf{S}_r^\top =\bDelta_r^{-\frac12}\bU_r^\top\bL^\top$, producing a standard balanced truncation reduced model defined as $\bA_r = \bS_r^\top\bA\bT_r$, $\bB_r = \bS_r^\top\bB$, $\bC_r = \bC\bT_r$.   The corresponding projection determined by the inference Gramians $\bH$ and $\gprior$ are given by $\widetilde{\mathbf{T}}_r=\bR\widetilde{\bZ}_r\bSigma_r^{-\frac12}$ and $\widetilde{\mathbf{S}}_r^\top =\bSigma_r^{-\frac12}\widetilde{\bU}_r^\top\widetilde{\bL}^\top$, producing a reduced model defined as $\widetilde{\bA}_r = \widetilde{\bS}_r^\top\bA\widetilde{\bT}_r$, $\widetilde{\bB}_r = \widetilde{\bS}_r^\top\bB$, $\widetilde{\bC}_r = \bC\widetilde{\bT}_r$.   Notice $\mathbf{V}_r=\bR\bZ_r=\bR\widetilde{\bZ}_r$, so that $\bZ_r=\widetilde{\bZ}_r$. Hence with $\bD=\bSigma_r^{-\frac12}\bDelta_r^{\frac12}$, we have
$$
\widetilde{\mathbf{T}}_r=\bR\widetilde{\bZ}_r\bSigma_r^{-\frac12}
=\bR\bZ_r\bDelta_r^{-\frac12}\bD=\mathbf{T}_r\bD
$$ 
Likewise, $\bL^\top\mathbf{V}_r=\bU_r\bDelta_r$, and $\widetilde{\bL}^\top\mathbf{V}_r=\widetilde{\bU}_r\bSigma_r$, which implies
$\bL\bU_r\bDelta_r=\Qm\mathbf{V}_r=\gprior^{-1}\mathbf{V}_r\bDelta_r^2$, and $\widetilde{\bL}\widetilde{\bU}_r\bSigma_r=\bH\mathbf{V}_r=\gprior^{-1}\mathbf{V}_r\bSigma_r^2$
so that 
$$
\widetilde{\mathbf{S}}_r^\top =\bSigma_r^{-\frac12}\widetilde{\bU}_r^\top\widetilde{\bL}^\top =\bSigma_r^{\frac12}\mathbf{V}_r^\top\gprior^{-1}
    =\bD^{-1}\bDelta_r^{\frac12}\mathbf{V}_r^\top\gprior^{-1}
    =\bD^{-1}\bDelta_r^{-\frac12}\bU_r^\top\bL^\top=\bD^{-1}\mathbf{S}_r^\top
$$
Thus, the two state space projections are equal: $\widetilde{\mathbf{T}}_r\widetilde{\mathbf{S}}_r^\top=\mathbf{T}_r\mathbf{S}_r^\top$ and the two reduced models are state-space equivalent: $\widetilde{\bA}_r=\bD^{-1}\bA_r\bD$, $\widetilde{\bB}_r=\bD^{-1}\bB_r$, and $\widetilde{\bC}_r=\bC_r\bD$. 
\end{proof}

\Cref{prop:Spantini=BT} tells us that when the BT-H and BT-Q generalized eigenproblems define the same dominant $r$ directions, then guarantees applicable to the BT-Q approach will transfer to the BT-H approach as well. We note that this result relies crucially on the fact that both approaches use $\gprior$ as the reachability Gramian -- matching of the dominant subspaces is not sufficient to guarantee state-space equivalent reduced models if different reachability Gramians are used. We numerically compare the BT-Q and BT-H posterior approximations in \Cref{sec: num exp}: our results show that while the BT-H approach often leads to similar results as the BT-Q approach, the BT-Q approach is more robust.

\section{Relationships between System-Theoretic and Bayesian Inference Gramians}
\label{sec: discussion}
The proposed approach in \Cref{sec:low-rank-models} hinges on two relationships between the canonical system-theoretic Gramians and the inference Gramians that we define in~\Cref{ssec: inference Gramians}.  We discuss now in more depth some theoretical and practical issues pertaining to these connections.   In \Cref{ssec:prior compatibility}, we discuss  \Cref{def: compatible prior} for prior compatibility in more detail, and provide two strategies for constructing priors with the desired compatibility properties which are also suitable for the Bayesian inference task.  In~\Cref{ssec: observability discussion}, we explore the relationship between the Fisher information matrix and system-theoretic observability Gramians, and describe two limiting settings in which the Fisher information matrix recovers the noisy observability Gramian up to a multiplicative constant.

\subsection{Prior covariance compatability}\label{ssec:prior compatibility}
As briefly discussed in \Cref{ssec: inference Gramians}, if a prior covariance matrix $\gprior$ satisfies \Cref{def: compatible prior}, then there exists some $\bB\in\R^{\ns\times \nin}$ for which $\gprior$ is exactly the reachability Gramian for the LTI system~\eqref{eq: LTI dyn sys}. 
We emphasize that not all covariance matrices will be prior-compatible with respect to dynamics determined by an independently assigned stable matrix, $\mathbf{A}$.  
One example of this takes $\gprior=\bI$  and $\bA$ to be any stable matrix whose numerical range extends into the right half plane, i.e., a matrix $\bA$ having eigenvalues in the complex open left half-plane but such that $\bA+\bA^\top$ has at least one \emph{positive} eigenvalue. With $\ns=2$, for example, $\mathbf{A}=\left[\begin{array}{cc} -1 & 3\\ 0 & -2\end{array}\right]$ would have this property.  
We now present two strategies for formulating compatible prior covariances: \Cref{sssec: spin up prior} describes how a prior covariance may naturally be defined as the result of a spin-up process if an input port matrix is available. \Cref{sssec: compatibility modification} describes how an incompatible prior covariance may be modified to to ensure compatibility.

\subsubsection{Spin-up}\label{sssec: spin up prior}

Given $\bA$ and a  matrix $\bB\in\R^{\ns \times \nin}$, a compatible prior may naturally be defined by stochastically exciting or ``spinning up'' the dynamical system \eqref{eq: dxdt is Ax} from $t=-\infty$ to $t=0$ as follows:
\begin{align}
    \difd \bx = \begin{cases}
\bA \bx \, \difd t + \bB \,\difd \bW(t), & t < 0, \\
        \bA \bx\, \difd t , & t \geq 0,
   \end{cases}
    \label{eq: white noise dyn sys}
\end{align}
where the usual control-theoretic input~\cite{fortmann1977introduction,baur2014model,datta2004numerical} has been replaced by white noise. Here the noise is represented by differentials $ \difd \bW (t)$, where $\bW$ is a vector of dimension $\nin$ containing independent copies of standard Brownian motion~\cite{cwg:hsm}. The matrix $ \bB $ acts as the port matrix distributing the white noise driving onto the state components.  

The stationary distribution of the system~\eqref{eq: white noise dyn sys} during the stochastic driving period has statistics  given in the following proposition.
\begin{proposition}\label{prop:statcov}
  At any finite time $ t \leq 0$, the state $ \bx (t) $ governed by \cref{eq: white noise dyn sys} is normally distributed with constant mean and covariance given by 
  \begin{align}\label{eq:statcov}
    \E[\bx(t)] = 0, \qquad
    \mathsf{cov}(\bx(t))=\E[\bx(t)\bx^\top(t)] = \int_0^\infty \expe^{\bA \tau} \bB \bB^\top  \expe^{\bA^\top \tau} \, \difd \tau 
  \end{align}
   Moreover, $\bGamma\equiv\cov(\bx(0))$ is a solution to the Lyapunov equation, \eqref{LyapEqns},
  \begin{align}\label{eq:statcov_MMContrGram}
    \bA \bGamma + \bGamma \bA^\top  = - \bB\bB^\top.
  \end{align}
\end{proposition}
We provide the proof in \Cref{app:proof}.
Equating this stationary distribution with the prior distribution~\eqref{eq: prior} on $ \bxin = \bx (0)$ gives the relations:
\begin{align}
    \mprior = 0, \qquad \gprior = \int_0^\infty \expe^{\bA t}\bB\bB^\top\expe^{\bA^\top t}\,\difd t.
    \label{eq: inference prior definitions}
\end{align}
This is consistent with the practice of ``spinning up'' atmospheric models so that their starting state (at least before data from observations are incorporated) reflects a plausible ``background covariance''~\cite{johns1997second,seferian2016inconsistent}. 
Thus, \cref{eq: inference prior definitions} provides a natural way of specifying a prior distribution satisfying \Cref{def: compatible prior} if an input port matrix is specified.  Indeed, \cite{redmann_balanced_2018} define the (infinite-time)  reachability Gramian for a stochastic system such as \cref{eq: white noise dyn sys} as the stationary covariance of the state.

\subsubsection{Compatibility modification}\label{sssec: compatibility modification}
We now suppose that we are given a prior covariance $\gprior$ that does \emph{not} satisfy \Cref{def: compatible prior} and provide a direct computational procedure that modifies the prior so that it is compatible. Suppose $\bGamma_0$ is a given positive definite matrix that we take as an initial approximation for a prior covariance matrix that fails to be prior-compatible, i.e., 
$$
  \mathbf{M}_0=\mathbf{A}\bGamma_0+\bGamma_0\mathbf{A}^\top \not \preceq \mathbf{0},
$$
so that, say, $\mathbf{M}_0$ is indefinite. If we write the spectral factorization of $\mathbf{M}_0$ as 
\begin{align*}
  \mathbf{M}_0=[\mathbf{U}_+\ \mathbf{U}_-]\begin{bmatrix}\Lambda_+ & \\ & \Lambda_- \end{bmatrix} \begin{bmatrix}\mathbf{U}_+^\top \\[1mm] \mathbf{U}_-^\top \end{bmatrix}, 
\end{align*}
  where $\Lambda_+=\Omega_+^2$ contains the strictly positive eigenvalues of $\mathbf{M}_0$ while $\Lambda_-$ contains the complementary nonpositive (i.e., negative and zero) eigenvalues of $\mathbf{M}_0$.  The \emph{nearest} negative semidefinite matrix to $\mathbf{M}_0$ (with respect to either the spectral norm or the Frobenius norm) is $\mathbf{M}_-=\mathbf{U}_-\, \Lambda_-\, \mathbf{U}_-^\top$, and the modification to $\bGamma_0$ that achieves this nearest negative semidefinite residual satisfies:
$$
  \mathbf{A}(\bGamma_0+\Delta)+(\bGamma_0+\Delta)\mathbf{A}^\top= \mathbf{U}_-\, \Lambda_-\, \mathbf{U}_-^\top \preceq \mathbf{0}
$$
or equivalently, 
\begin{align}\label{eq:Deltalyap}
  \mathbf{A}\Delta+\Delta\mathbf{A}^\top + \mathbf{U}_+\, \Lambda_+\, \mathbf{U}_+^\top =0, 
\end{align}
which will have a unique positive definite solution, $\Delta=\mathbf{E}^\top\mathbf{E}$. Since $\mathbf{U}_+ \Lambda_+ \mathbf{U}_+^\top$ has an immediately accessible square root, namely $\mathbf{U}_+ \Omega_+$, the Cholesky factor, $\mathbf{E}$, of $\Delta$ can be computed directly using an approach introduced by Hammarling in \cite{hammarling1982lyapchol}. This in particular allows one to compute a unique compatible prior $\gprior \coloneqq \bGamma_0 + \Delta$. The modification, $\Delta$, to $\bGamma_0$ is \emph{not} generally the smallest modification to $\bGamma_0$ making it prior-compatible. Rather, it is the uniquely determined modification to $\bGamma_0$ producing the nearest negative semidefinite residual. See \Cref{app:ModPrior} for a routine in \textsc{matlab} that implements this compatibility computation.

\subsection{Observability Gramians and the Fisher information matrix}\label{ssec: observability discussion}
Our proposed approach is motivated by connections between the Fisher information matrix and system-theoretic observability Gramians.
In contrast to the connection between the prior covariance and the reachability Gramian discussed in~\Cref{ssec:prior compatibility}, we describe in \Cref{sssec: Fisher incompat} a fundamental obstacle to establishing a precise relation between the Fisher information matrix and an observability Gramian under a general given observation protocol. 
However, the Fisher information matrix $\bH$ can approach the noisy observability Gramian $\Qm$, up to a scaling factor of sampling frequency, in certain limits. Note that this scaling factor of inverse time arises because $\Qm$ is a time integral and thus has an extra dimension of time when compared to $\bH$, which is a sum. \Cref{sec:contlim} considers the limit of continuous measurements and \Cref{sec:randgram} considers infinite discrete-time measurements taken at times distributed according to a Poisson process.
These approaches are related to asymptotics of covariance approximations by Fisher information matrices \cite{abt_fisher_1998}, and are also related to models of infill and domain asymptotics in statistics \cite{chen_infill_2000,zhang_towards_2005,zhu_spatial_2006}.  

While infinite-horizon definitions~\eqref{eq: standard infinite gramians} of system-theoretic Gramians are most commonly used, and indeed are required for the results of \Cref{lemma: continuous time BT,lem: BT Goodness} to hold, other Gramian definitions exist~\cite{antoulas2005approximation,gugercin2004survey,BennerB2017}. For example, one can define the Gramians over a finite time interval $[t_1,t_2]$, leading to time-limited Gramians. Furthermore, in practical balanced truncation settings, the infinite time-horizon Gramians are sometimes approximated via numerical quadrature~ \cite{moore1981principal,WP02,Ro05,BF19,breiten2016structure} or via approximate solution methods to the Lyapunov equations that yield low-rank Gramian approximations~\cite{simoncini2016computational,kurschner2016efficient,BenS13}. 
In this context, the scaled Fisher information matrix can be viewed as an approximation to the infinite noisy Gramian of \Cref{def: noisy obs Gramian}, see, e.g., \cite{powel_empirical_2015,roy_relation_2009}. When $\bH$ is a sufficiently close approximation to $\Qm$ so that their dominant generalized eigenpairs with $\gprior^{-1}$ coincide, then \Cref{prop: QmBT is hatH} applies.

\subsubsection{Incompatibility of the Fisher information matrix}\label{sssec: Fisher incompat}

From formal duality, one could formulate a \emph{likelihood-compatibility} condition, in analogy to \Cref{def: compatible prior} for the prior covariance, by saying 
a Fisher information matrix $\bH$ is \emph{likelihood-compatible} with state dynamics induced by $\bA$, if $\bH\bA + \bA^\top\bH$ is negative semidefinite.  That would imply the existence of some $\bC\in\R^{\nout \times \ns}$ such that $\bH$ solves a dual Lyapunov equation~\eqref{LyapEqns} $\bA^\top\bH + \bH\bA = - \bC^\top\bC$, and is thus the infinite-horizon observability Gramian of an LTI dynamical system~\eqref{eq: LTI dyn sys} with output port matrix $\bC$. Replacing the noisy observability Gramian $\Qm$ with a likelihood-compatible $\bH$ in the BT-Q approach of \Cref{ssec: method} (which requires a compatible prior) would then yield a reduced model that inherits the theoretical guarantees of \Cref{lem: BT Goodness} and \Cref{prop: QmBT is hatH}.

Unfortunately, Fisher information matrices for observation models that occur in plausible inference settings will be extremely unlikely to be likelihood-compatible. In particular, the Fisher information as defined in \eqref{eq: fisher info sum}, $\mathbf{H}=\sum_{i=1}^n e^{\mathbf{A}^Tt_i}\mathbf{C}^T\gobs^{-1}\mathbf{C}e^{\mathbf{A}t_i}$,  is often rank-deficient, and in that case in order for $\mathbf{H}$ also to be likelihood compatible with respect to the dynamics determined by $\bA$, it is necessary for $\mathsf{ker}(\mathbf{H})$ to be an invariant subspace for $\mathbf{A}$, which is extraordinarily unlikely to occur with measurement models specified in practice.
One may include adjustments to the Fisher information intended to enforce likelihood-compatibility (analogously to the prior compatibility modification of \Cref{sssec: compatibility modification}) however such adjustments appear difficult to justify since they would represent significant \emph{ad hoc} changes to the observation protocol that at best would be unrealistic. 

\subsubsection{Continuum Limit of Regularly Spaced Observations}
\label{sec:contlim}
We proceed next to show that despite the frustrated equivalence between the Fisher information matrix and a continuous-time observability Gramian under general observation protocols, we can establish a formal proportionality between these quantities in an idealized continuum limit of finely spaced perpetual observations.
Suppose the measurement times $\{t_i\}_{i=1}^\nmeas$, with $\nmeas\geq 2$, in the definition of the Fisher information matrix~\eqref{eq: fisher info sum} are equispaced  in the interval
$[\tstart, \tend]$, i.e.,
\begin{align}
  t_i &\coloneqq \tstart + (i-1) h, & h &\coloneqq \frac{\tend - \tstart}{\nmeas-1}, \label{eq:regsamp}
\end{align}
for $i = 1, \ldots, \nmeas$. 
If we scale the Fisher information matrix by the sampling time interval $ h$,
then the large-$\nmeas$ energy limits to an integral (we make the dependence of $\bH$ on $\nmeas$ explicit here):
\begin{align}\label{eq: fisher info infinity}
    \hat{\bH}_\infty \coloneqq \lim_{n\to\infty}h\,\bH_\nmeas = \lim_{n\to\infty}h\sum_{i=1}^n \expe^{\bA^\top t_i}\bC^\top \gmeas^{-1}\bC\expe^{\bA t_i}= \int_{\tstart}^{\tend}\expe^{\bA^\top t}\bC^\top \gmeas^{-1}\bC\expe^{\bA t}\,\difd t.
\end{align}
In this finite observation interval setting, $\hat{\bH}_\infty$ satisfies the modified Lyapunov equation,
\begin{equation}\label{eq:modified-lyapunov}
  \bA^\top\hat{\bH}_\infty + \hat{\bH}_\infty\bA =e^{\bA^\top \tend} \bC^\top \gmeas^{-1} \bC e^{\bA \tend}-e^{\bA^\top \tstart} \bC^\top \gmeas^{-1} \bC e^{\bA \tstart}.
\end{equation}

We recover the noisy infinite observability Gramian of \Cref{def: noisy obs Gramian} if we take the observation start time $ \tstart \downarrow 0$ and observation end time $ \tend \uparrow \infty$ together with the finer spacing $\rev{h} $ of the observations, provided the normalization of the energy scales appropriately with the sampling interval.
\begin{proposition}\label{prop:continuum-limit}
  Let a stable matrix $\bA$ and an output matrix $\bC$ be given. Define a collection of observation times $ \{t_i\}_{i=1}^{\nmeas}$ by $ t_i = (i-1)h(\nmeas) $ where $h(\nmeas) \equiv \frac{T(\nmeas)}{\nmeas-1}$ with the observation time interval $ T (\nmeas) $ increasing sublinearly with the number of observations so that $ T (\nmeas)\rightarrow \infty  $ while $ h(\nmeas)= \frac{T(\nmeas)}{\nmeas-1} \rightarrow 0$ as $ \nmeas \rightarrow \infty $.
  Then,
  \begin{align}\label{eq:continuum-limit}
  \lim_{\nmeas \rightarrow \infty}
  h(n) \sum_{i=1}^{\nmeas}  e^{\bA^\top t_i} \bC^\top \gmeas^{-1} \bC e^{\bA t_i}  = \int_0^\infty e^{\bA^\top t} \bC^\top \gmeas^{-1} \bC e^{\bA t} \difd t.
  \end{align}
\end{proposition}

{\begin{proof}
The integrand on the right hand side and the summand on the left hand side are continuous matrix functions whose norm is bounded and decays exponentially.  Truncating both the sum and the integral to times $ 0 \leq t, t_i \leq \tau$ leads to a Riemann sum convergence on this truncated interval.  The truncated component $ t_i, t >\tau$ can be bounded in norm by a function of the form $ k_1 \expe^{-k_2 \tau} $ for positive constants $k_1$ and $k_2$ depending on the norm of $ \bC$, the norm bound on $ \gmeas^{-1}$, and the stability properties of the matrix $ \bA$.  By taking $ \tau \uparrow \infty$, we achieve the stated result. \hfill$\Box$
\end{proof}}

\subsubsection{Observation Models at Random Times}
\label{sec:randgram}
Observability Gramians with an integral form can also be associated to observation models that take place at discrete, but randomly distributed, moments of time. For example, the observation times in a field experiment might not be perfectly predictable due to considerations of weather, availability of human resources, and/or availability of equipment.  We do not propose randomly distributed observation times as an important practical innovation, but rather to provide another means of interpreting the relation between the observability Gramian and the Fisher information.

As in the previous section, we specify a time $ \tstart$ when observations begin and a time $ \tend$ when observations end.  If we have $ \nmeas $ observations to be taken at independent and uniformly distributed times over this interval $ t_i \sim U(\tstart,\tend)$, we would have a random Fisher information due to the randomness in the $t_i$.  However, if we average over the realizations of the uniformly distributed observation times, we would obtain the mean Fisher information:
\begin{align*}
    \bar\bH &\equiv (\tend-\tstart)^{-\nmeas} \int_{\tstart}^{\tend} \int_{\tstart}^{\tend} \cdots \int_{\tstart}^{\tend} 
    \sum_{i=1}^{\nmeas}  e^{\bA^\top t_i^{\prime}} \bC^\top {\gmeas^{-1}} \bC e^{\bA t_i^{\prime}} \, \difd t_1^{\prime} \, \difd t_2^{\prime} \ldots \difd t_{\nmeas}^{\prime} \\
    &= \frac{\nmeas}{\tend-\tstart} \int_{\tstart}^{\tend} 
    e^{\bA^\top t} \bC^\top \gmeas^{-1} \bC e^{\bA t} \, \difd t.
\end{align*}
The same mean Fisher information would result if we divided the observation period into $ \nmeas$ equal subintervals of length $ h= (\tend-\tstart)/\nmeas $, and took one observation at a time $t_i$ that is uniformly distributed within each size-$h$ subinterval:
\begin{align}\nonumber
  \bar\bH &= \sum_{i=1}^n \frac{1}{h}\int_{\tstart +(i-1)h/\nmeas}^{\tstart+ih/\nmeas} e^{\bA^\top t_i^{\prime}} \bC^\top \gmeas^{-1} \bC e^{\bA t_i^{\prime}} \, \difd t_i^{\prime} \\
    &= \frac{1}{h} \int_{\tstart}^{\tend} 
    e^{\bA^\top t} \bC^\top \gmeas^{-1}\bC e^{\bA t} \, \difd t.\label{eq: mean Fisher info}
\end{align}
Notice we have $ \nmeas $ subintervals in this formulation, as compared to the $ \nmeas-1 $ subintervals in \Cref{sec:contlim} because we have here $ \nmeas$ observations distributed in the interior, rather than the endpoints, of the subintervals.

To obtain the standard infinite horizon observability Gramian~\eqref{eq: standard infinite gramians} under a similar paradigm, we consider a countably infinite sequence of observations at randomly distributed times ($\nmeas \rightarrow \infty$). Under the appropriate limiting models of random measurements we can recover, in expectation, an infinite-horizon Gramian multiplied by the sampling frequency.
\begin{proposition}\label{prop:mean-observability-limit}
  Let a stable matrix $\bA$ and an output matrix $\bC$ be given. Let $\{t_i\}_{i \geq 1}$ be a random sequence of times defined by either of the following models:
  \begin{enumerate}
    \item[(a)] For some $h > 0$, $t_i \sim U((i-1)h, i h)$. 
    \item[(b)] For some $h > 0$, {the $\{t_i\}_{i\geq 1 }$ are the realization of a Poisson point process on $ [0,\infty) $ with intensity $h^{-1}$.}
  \end{enumerate}
  Then, as $n \uparrow \infty$, the expected Fisher information~\eqref{eq: mean Fisher info} limits {to an  infinite-time observability Gramian scaled by the mean sampling frequency:}
  \begin{align}\label{eq:mean-observability-limit}
    \lim_{\nmeas \rightarrow \infty} \E\left[ \sum_{i=1}^{\nmeas}  e^{\bA^\top t_i} \bC^\top \gmeas^{-1} \bC e^{\bA t_i} \right] = h^{-1} \int_0^\infty e^{\bA^\top t} \bC^\top \gmeas^{-1} \bC e^{\bA t} \difd t.
  \end{align}
\end{proposition}
\begin{proof}
  {Define
  \begin{align*}
    f(t) \coloneqq \bxi^\top \cdot e^{\bA^\top t} \bC^\top \gmeas^{-1} \bC e^{\bA t} \cdot \bxi.
  \end{align*}
  where $ \bxi \in \mathbb{R}^{d}$.}
  Under model (a), a similar computation as in \cref{eq: mean Fisher info} yields,
  \begin{align}
    \E \left[ \sum_{i=1}^{\nmeas} f(t_i) \right] = h^{-1} \int_0^{h n} f(t) \difd t.
  \end{align}
  Since $f$ is {non-negative}, continuous for all $t \geq 0$, and decays exponentially for large $t$, we also have,
  \begin{align}
    \lim_{n \rightarrow \infty} \int_0^{h n} f(t) \difd t = \int_0^\infty f(t) \difd t.
  \end{align}
  As $ \bxi\in \mathbb{R}^d$ was arbitrary, we obtain the result~\eqref{eq:mean-observability-limit}.
{
  Under model (b), note that $f$ is integrable for $t \geq 0$. Then Campbell's theorem for Poisson processes asserts that 
  \begin{align*}
    \mathbb{E} [\sum_{i=1}^{\infty} f(t_i)] = h^{-1} \int_0^{\infty} f(t) \, \difd t.
  \end{align*}
This, along with, 
    \begin{align}
   \E\left[ \lim_{\nmeas \rightarrow \infty} \sum_{i=1}^{\nmeas} f(t_i) \right] = \lim_{\nmeas \rightarrow \infty} \E \left[ \sum_{i=1}^{\nmeas} f(t_i) \right],
  \end{align}
  which follows from the monotone convergence theorem, and again the arbitrary choice of $ \bxi \in \mathbb{R}^d$, proves the desired result.}\hfill$\Box$
\end{proof}

\noindent
\emph{Remark.} The standard observability Gramian~\eqref{eq: standard infinite gramians} would correspond to an infinite sequence of observations at unit frequency, either according to a standard Poisson point process or a randomized sampling within regularly spaced intervals. The above result also implies that the large-$n$ expectation of the Fisher information matrix \eqref{eq: fisher info sum} equals the (scaled) infinite-time observability Gramian, assuming the observations $t_i$ are distributed according to one of the models in \Cref{prop:mean-observability-limit}.

\section{Numerical experiments}\label{sec: num exp}
We now demonstrate the balanced truncation approaches on two test problems based on LTI systems that are benchmarks within the system-theoretic model reduction community\footnote{LTI system matrices and documentation for both examples can be found online at {\tt http://slicot.org/20-site/126-benchmark-examples-for-model-reduction}.}: a heat equation example and a structural model for the ISS1R module.  For each LTI system, the inference problem is to infer the initial system state from output measurements made after $t=0$ corrupted by Gaussian noise. For each system, we will use the system dynamics to determine a compatible prior covariance matrix as defined in \Cref{def: compatible prior}. Then, we will compare the posterior mean and covariance approximation performance of three approaches:
\begin{enumerate}
    \item from~\cite{spantini_optimal_2015}, the optimal low-rank update (\textbf{OLRU}) approximations of the posterior covariance~\eqref{truncK} and mean~\eqref{eq: low rank update mean}, and the optimal low-rank (\textbf{OLR}) posterior mean approximation,
    \item the \textbf{BT-Q} posterior mean and covariance approximation of \Cref{ssec: method} based on the pencil $(\Qm, \gprior^{-1})$, and
    \item the \textbf{BT-H} posterior mean and covariance approximation of \Cref{ssec: method} based on the pencil $(\bH,\gprior^{-1})$.
\end{enumerate}
For both LTI systems, we consider both a `good' observation model where the final observation time, $T$, and number of measurements, $n$, are both large, as well as a `bad' observation model where $T$ and $n$ are both small.
\Cref{ssec: heat results,ssec: ISS results} describe the heat equation and ISS1R problems, respectively, and present numerical results for the aforementioned posterior approximations. \Cref{ssec: numerical discussion} discusses the numerical results from both examples. 

\subsection{Heat equation}\label{ssec: heat results}
We consider the heat equation in a one-dimensional rod with homogeneous Dirichlet boundary conditions. The state and output dimensions are $\ns=200$ and $\nout=1$, respectively, and the single output is the temperature measured at 2/3 of the rod length. We define the prior covariance $\gprior$ to be the infinite reachability Gramian~\eqref{eq: standard infinite gramians} that would result from taking $\bB=\bI$. The true initial condition is drawn from $\cN(0,\gprior)$ and used to generate noisy output measurements at $t_i=ih$, for $i = 1,2,\ldots,n$. The noise values at each measurement are independently and identically distributed: $\gmeasi=\gmeas= \sigma_{\rm obs}^2$, where the noise standard deviation is chosen to be roughly 10\% of maximum magnitude of the noiseless output, $\sigma_{\rm obs}=0.008$.

\begin{figure}[h]
    \centering
    \includegraphics[width=0.32\textwidth]{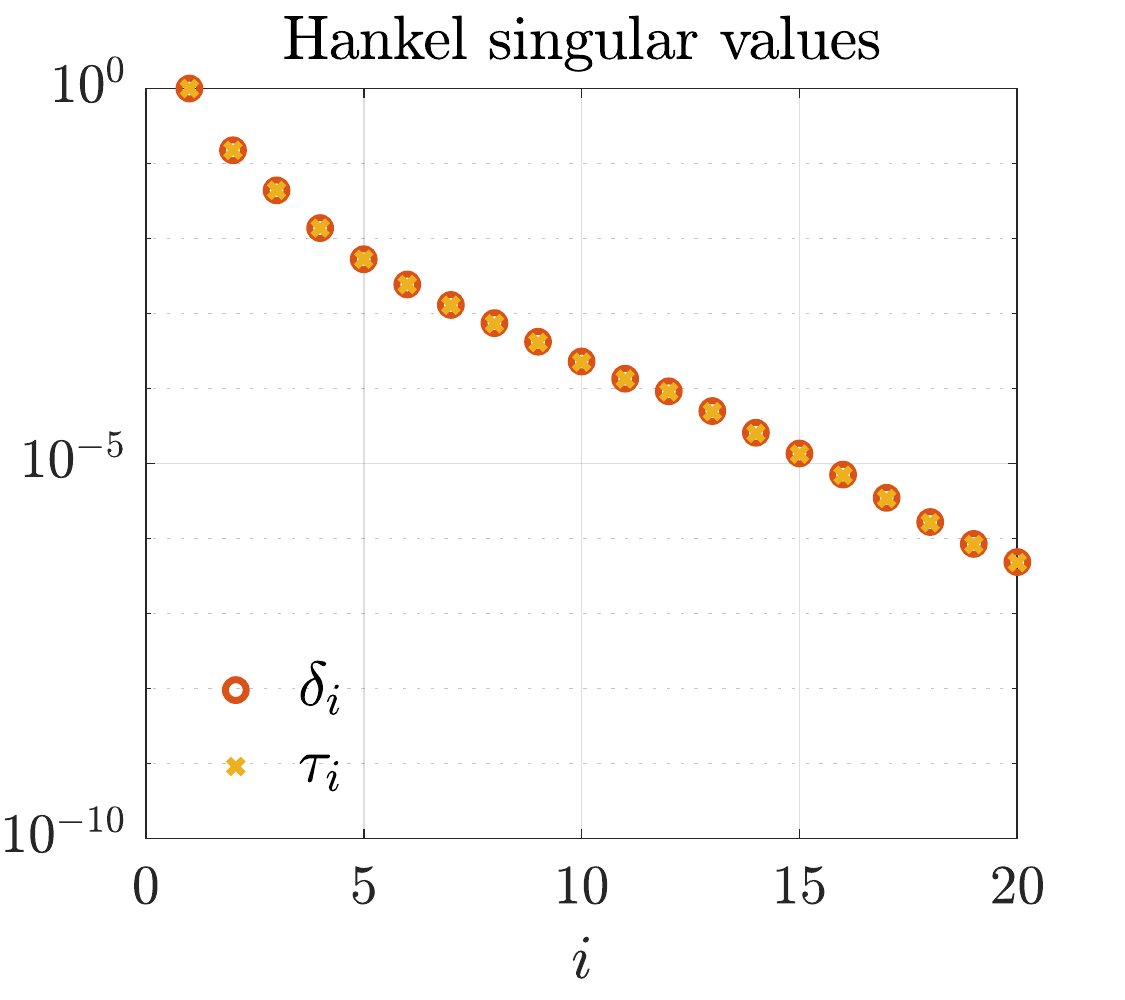}
    \hfill
    \includegraphics[width=0.32\textwidth]{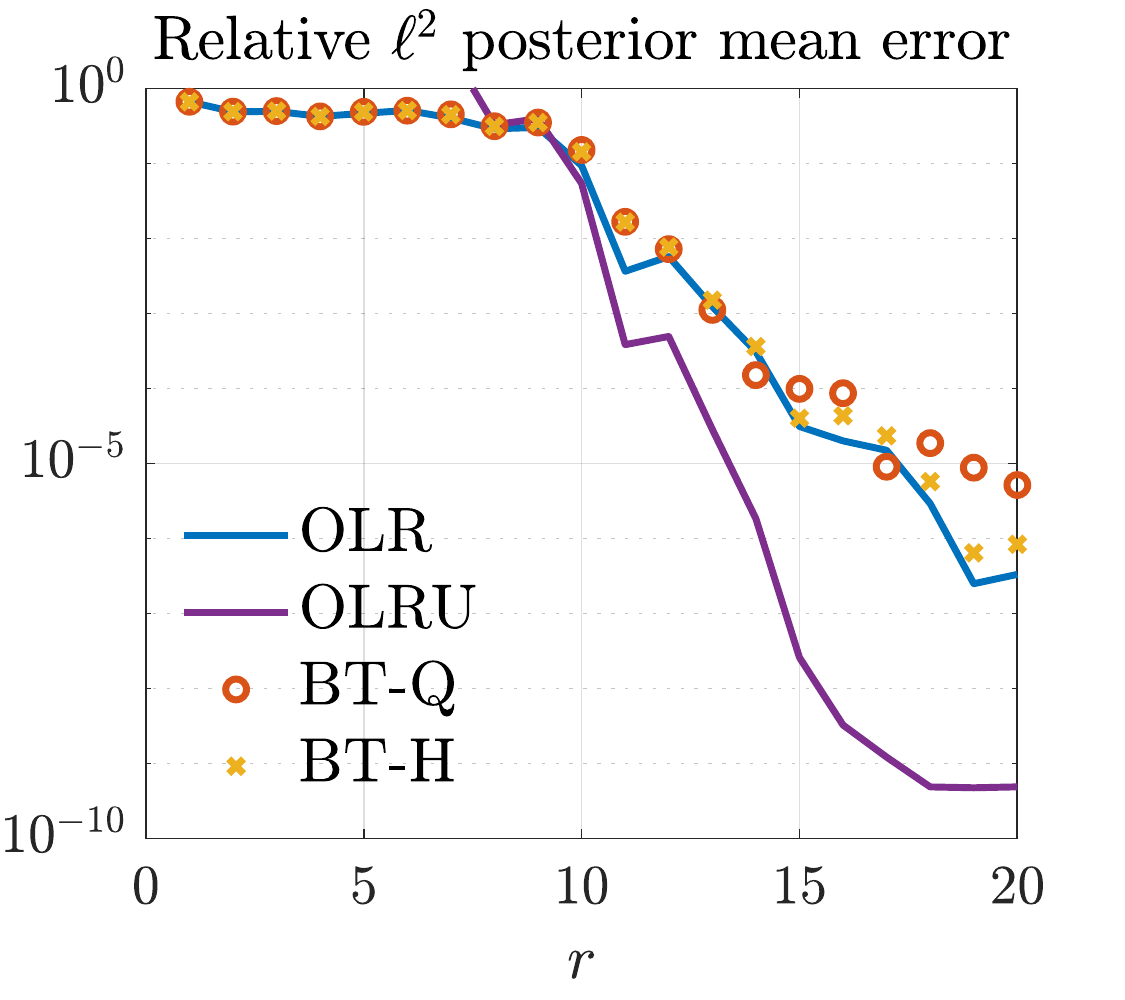}
    \hfill
    \includegraphics[width=0.32\textwidth]{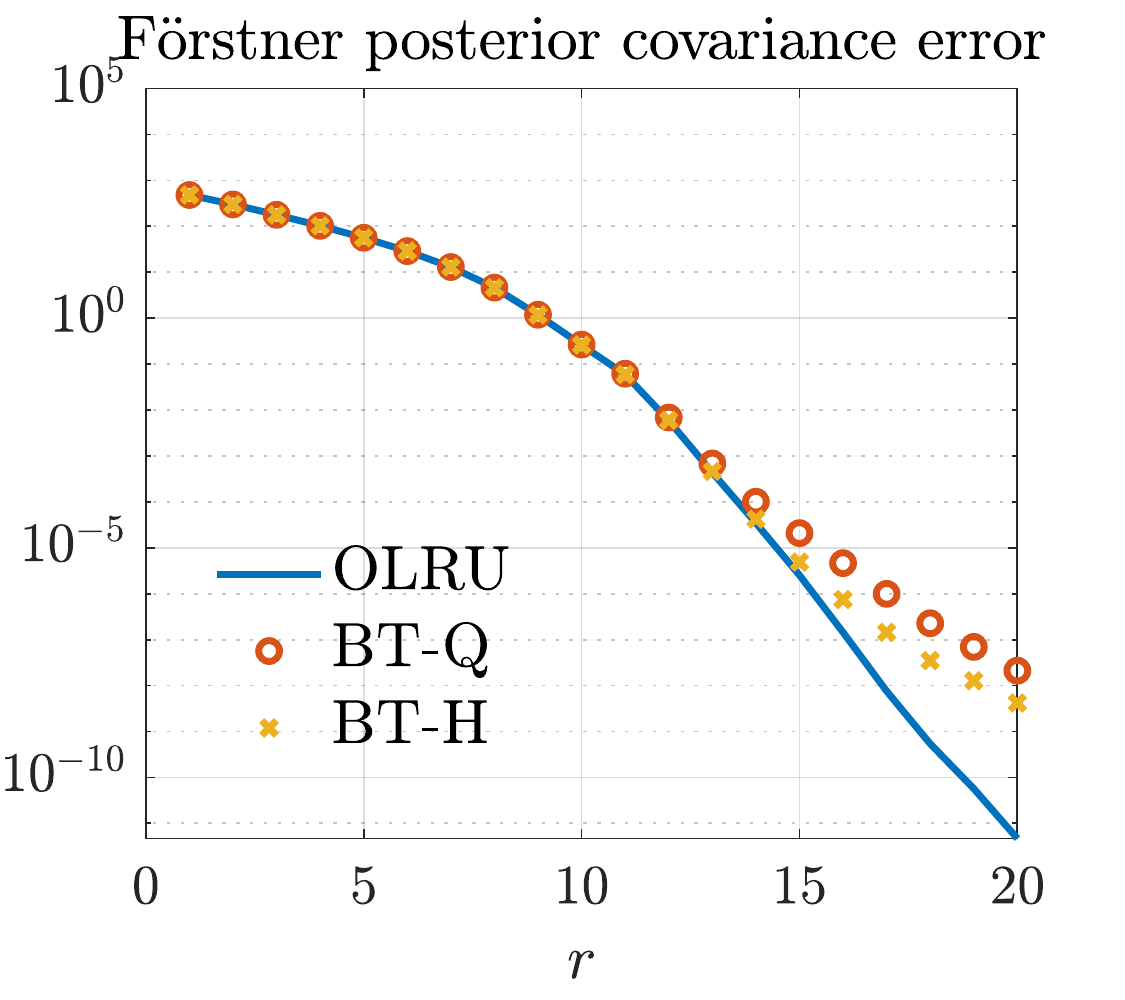}
    \caption{Results for the heat equation model with measurements spaced $h=10^{-4}$ apart until $T=50$, for which $\Qm$ and $h\bH$ have a relative error of 0.1\%. In the left panel we plot $ \tau_i $ and $ \delta_i $, the square roots of the generalized eigenvalues of the pencils $(\bH,\gprior^{-1})$ and $(\Qm,\gprior^{-1})$, respectively. These are normalized relative to $\delta_1, \tau_1$, respectively, and coincide with the Hankel singular values.  In the middle and right panels, we compare respectively the relative errors of the posterior means of the various low-rank approximation approaches in the $ \ell^2$ norm, and the errors of the posterior covariances in the F\"{o}rstner distance.
    }
    \label{fig:heat good}
\end{figure}

\begin{figure}[h]
    \centering
    \includegraphics[width=0.32\textwidth]{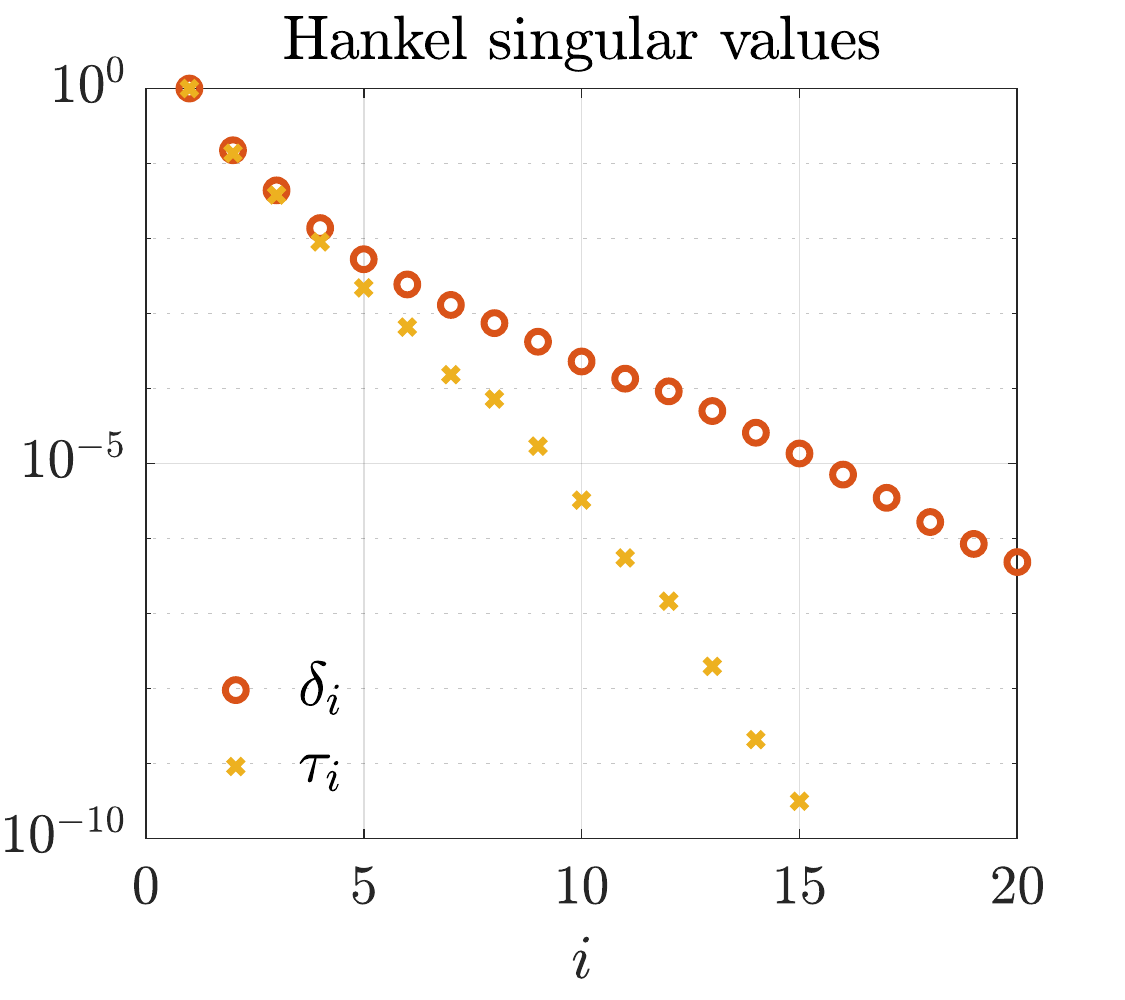}
    \hfill
    \includegraphics[width=0.32\textwidth]{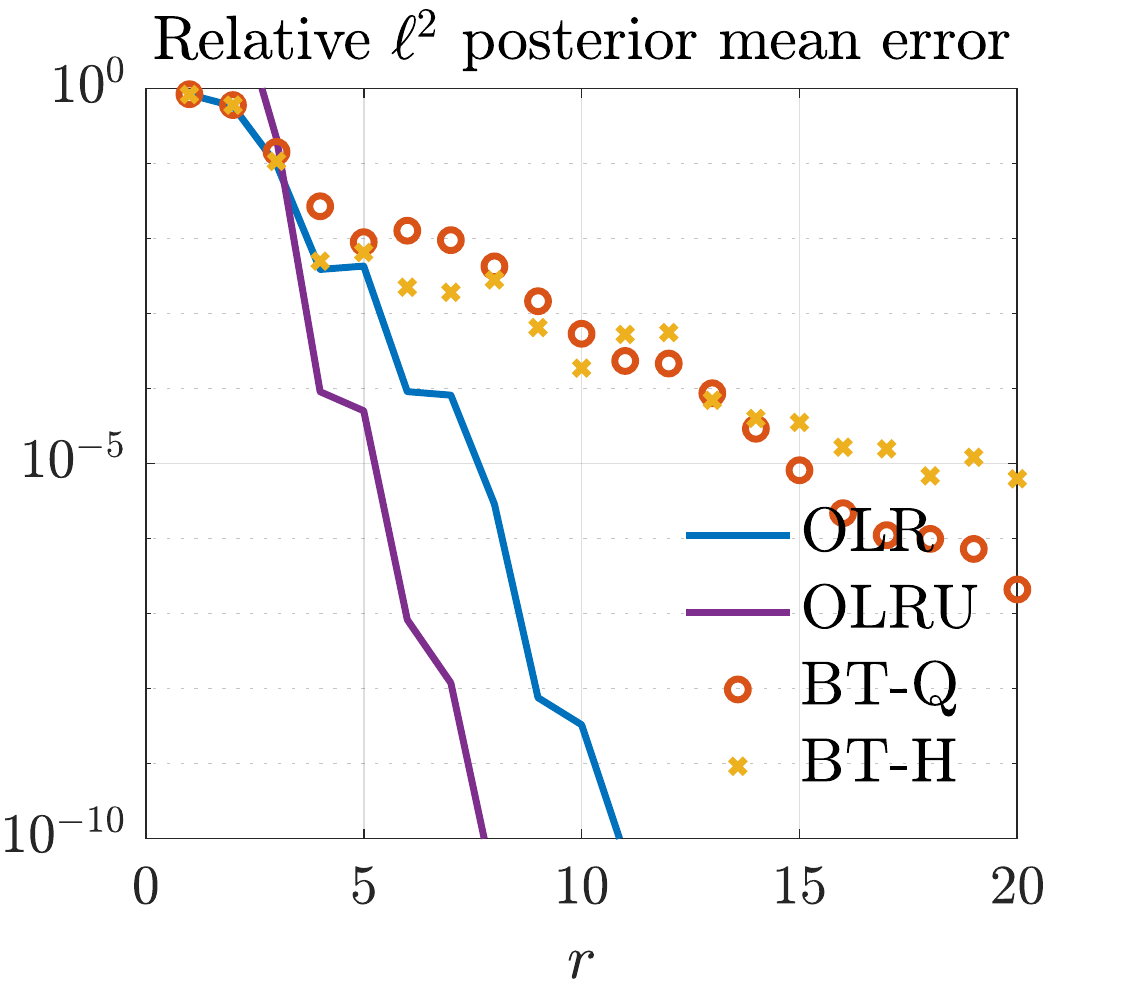}
    \hfill
    \includegraphics[width=0.32\textwidth]{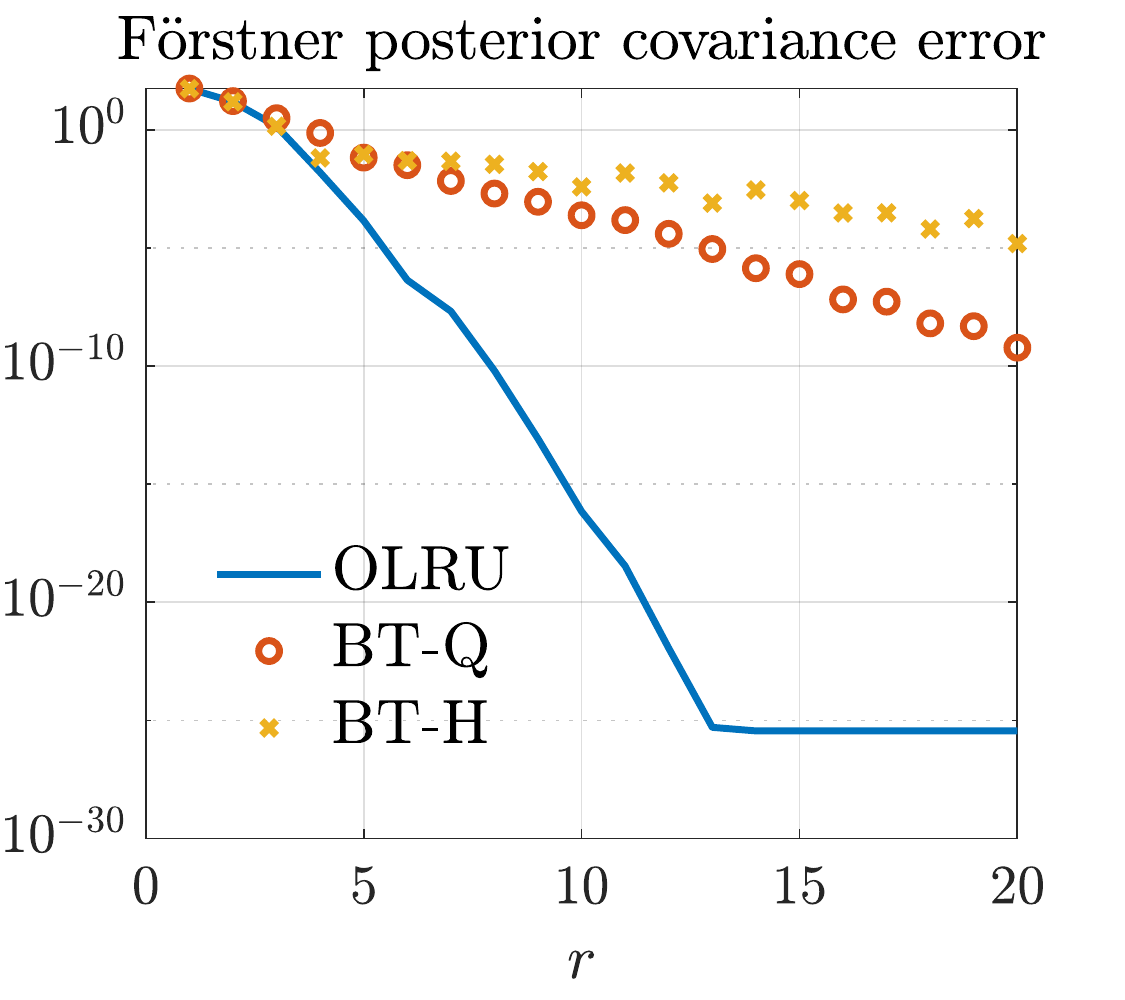}
    \caption{Results for the heat equation model with measurements spaced $h=10^{-1}$ apart until $T=10$, for which $\Qm$ and $h\bH$ have a relative error of 15\%. The quantities plotted are as defined in Figure~\ref{fig:heat good}.
    }
    \label{fig:heat bad}
\end{figure}

We consider in \Cref{fig:heat good} an observation model with $n=5 \times 10^5$ measurements spaced $h = 10^{-4}$ apart, for a final observation time $T = 50$. The spectral abscissa of $\bA$ for this system is $\approx -0.1$, so this $T$ value is about five times the slowest characteristic timescale. We therefore expect that the finely spaced observations in this model can capture most of the system dynamics and approximate the infinite integral. Indeed, the relative Frobenius norm difference between $h\bH$ and $\bQ$ is 0.1\% in this case. In contrast, \Cref{fig:heat bad} considers a limited observation model with $n = 100$ measurements spaced $h = 0.1$ apart, for a final observation time of $T = 10$. In this case the relative Frobenius norm difference between $h\bH$ and $\bQ$ is 15\%.

\subsection{ISS1R module}\label{ssec: ISS results}
This benchmark LTI system is a structural model for the flex modes of the Zvezna service module of the International Space Station~\cite{antoulas2005approximation}. The state and output dimensions are $\ns=270$ and $\nout=3$, respectively; the outputs correspond to roll, pitch, and yaw gyroscope measurements. An input matrix $\bB$ with $\nin=3$ is provided as part of the model that corresponds to roll, pitch, and yaw jet controls. This $\bB$ is used to generate a compatible prior $\gprior$. The true initial condition is drawn from $\cN(0,\gprior)$ and used to generate noisy output measurements. The noise covariance $\gmeas=\diag\{0.0025^2, 0.0005^2, 0.0005^2\}$ is set so that the noise standard deviation of each output is roughly 10\% of the maximum output signal as before.

\Cref{fig:iss good} considers an observation model with $n = 3000$ measurements spaced $h= 0.1$ apart for a final observation time $T = 300$. In this case, the relative difference between $h\bH$ and $\Qm$ is 1\%. \Cref{fig:iss bad} considers $n = 10$ measurements spaced $h = 1$ apart for a final observation time $T = 10$. In this case, the relative difference between $h\bH$ and $\Qm$ is 53\%.

\begin{figure}[h]
    \centering
    \includegraphics[width=0.32\textwidth]{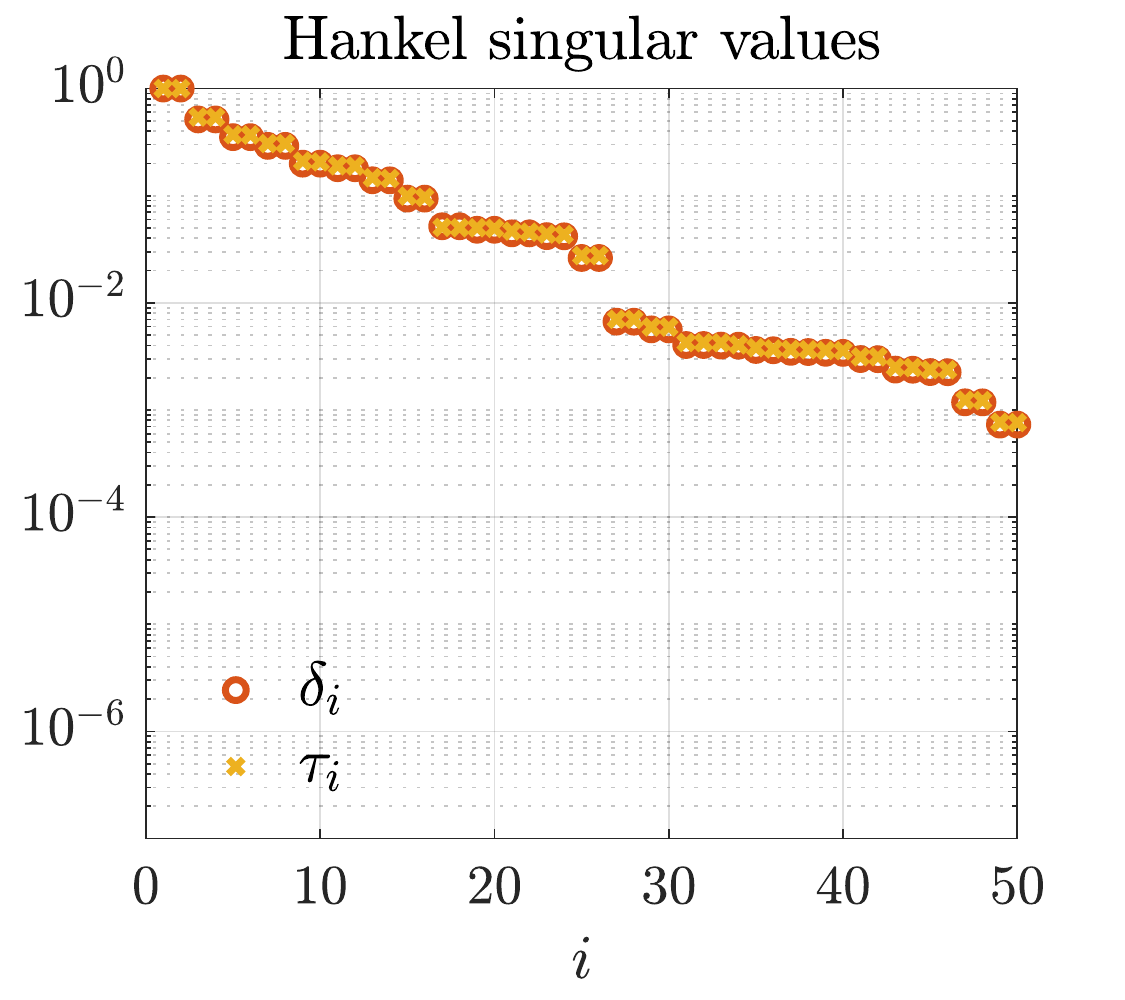}
    \hfill
    \includegraphics[width=0.32\textwidth]{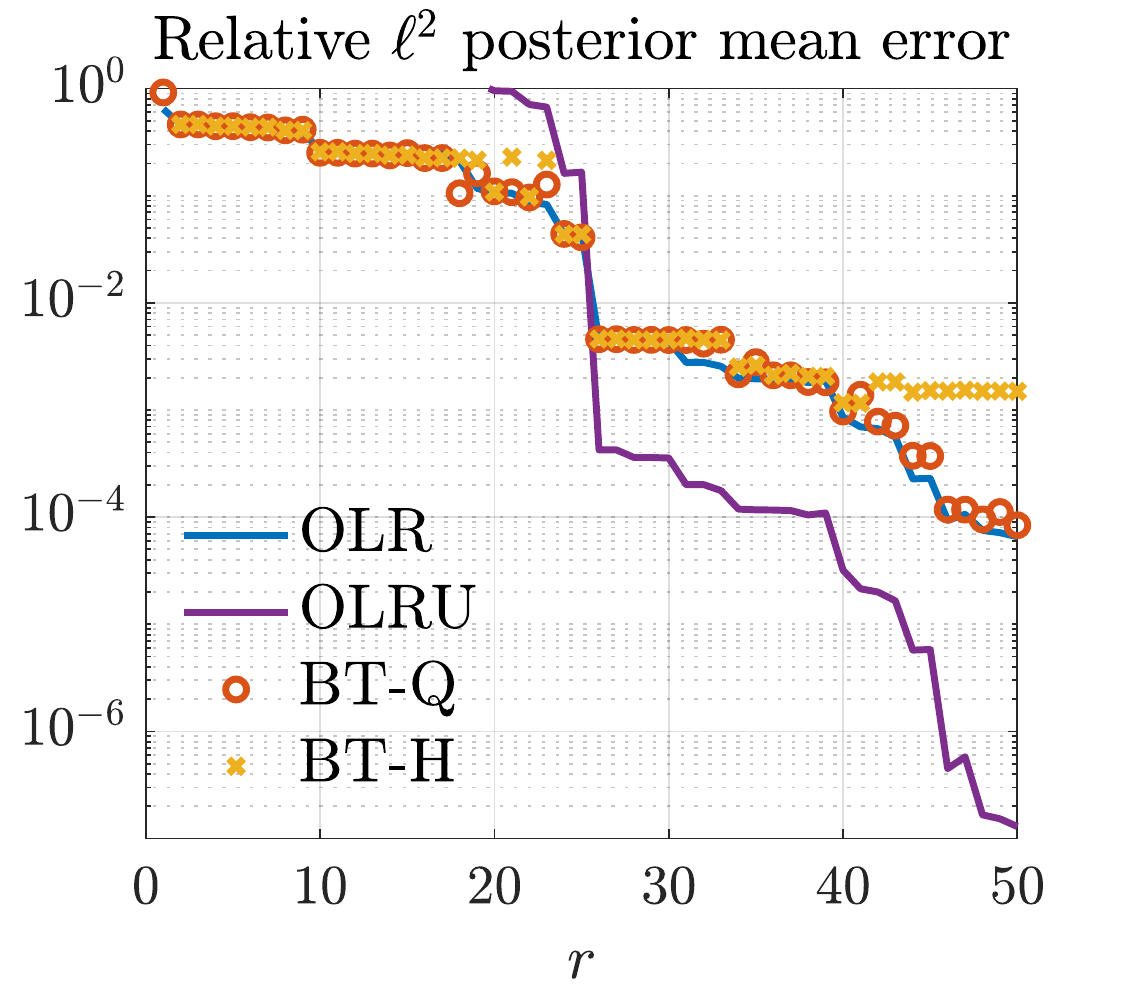}
    \hfill
    \includegraphics[width=0.32\textwidth]{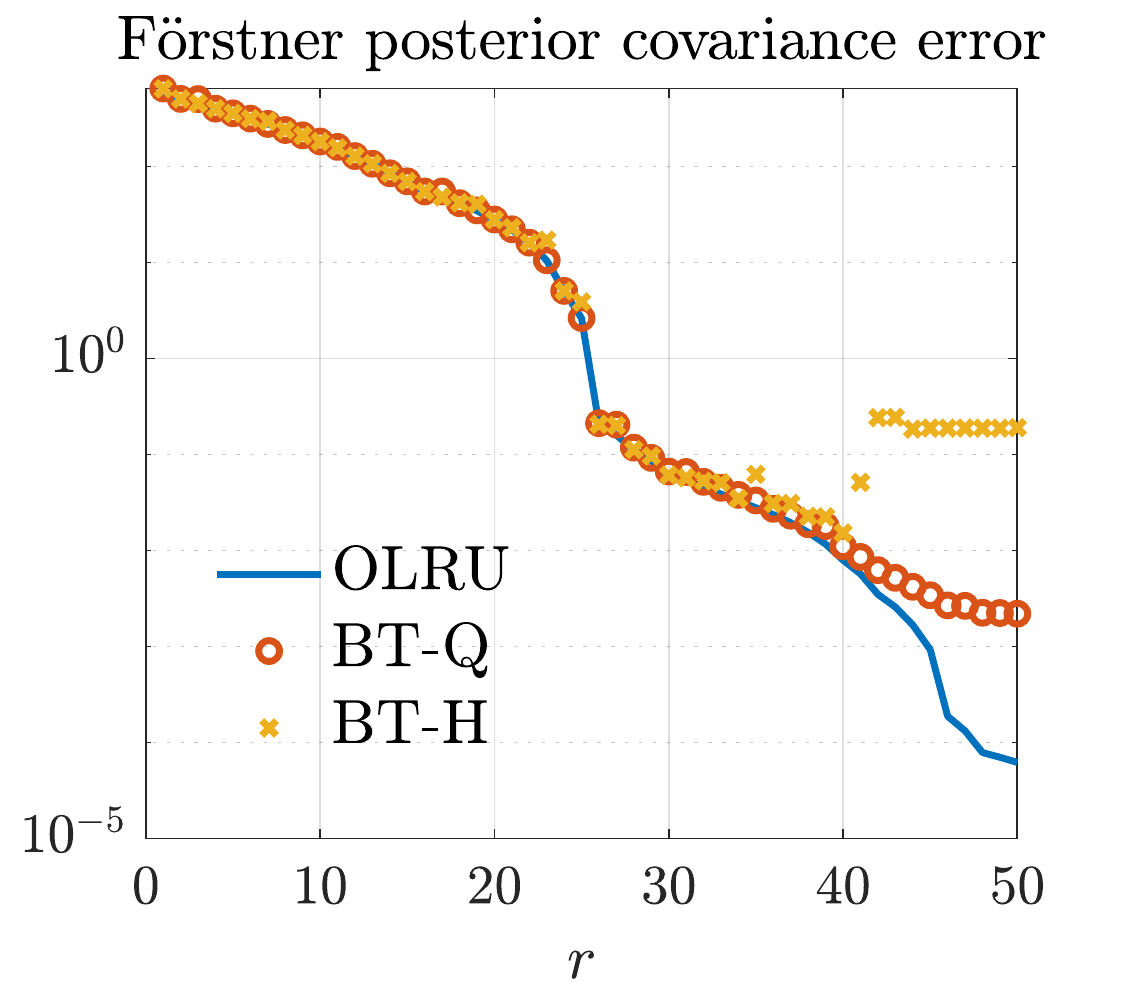}
    \caption{Results for the ISS1R model with measurements spaced $h=0.1$ apart until $T=300$, for which $\Qm$ and $h\bH$ have a relative error of 1\%. The quantities plotted are as defined in Figure~\ref{fig:heat good}.
    }
    \label{fig:iss good}
\end{figure}

\begin{figure}[h]
    \centering
    \includegraphics[width=0.32\textwidth]{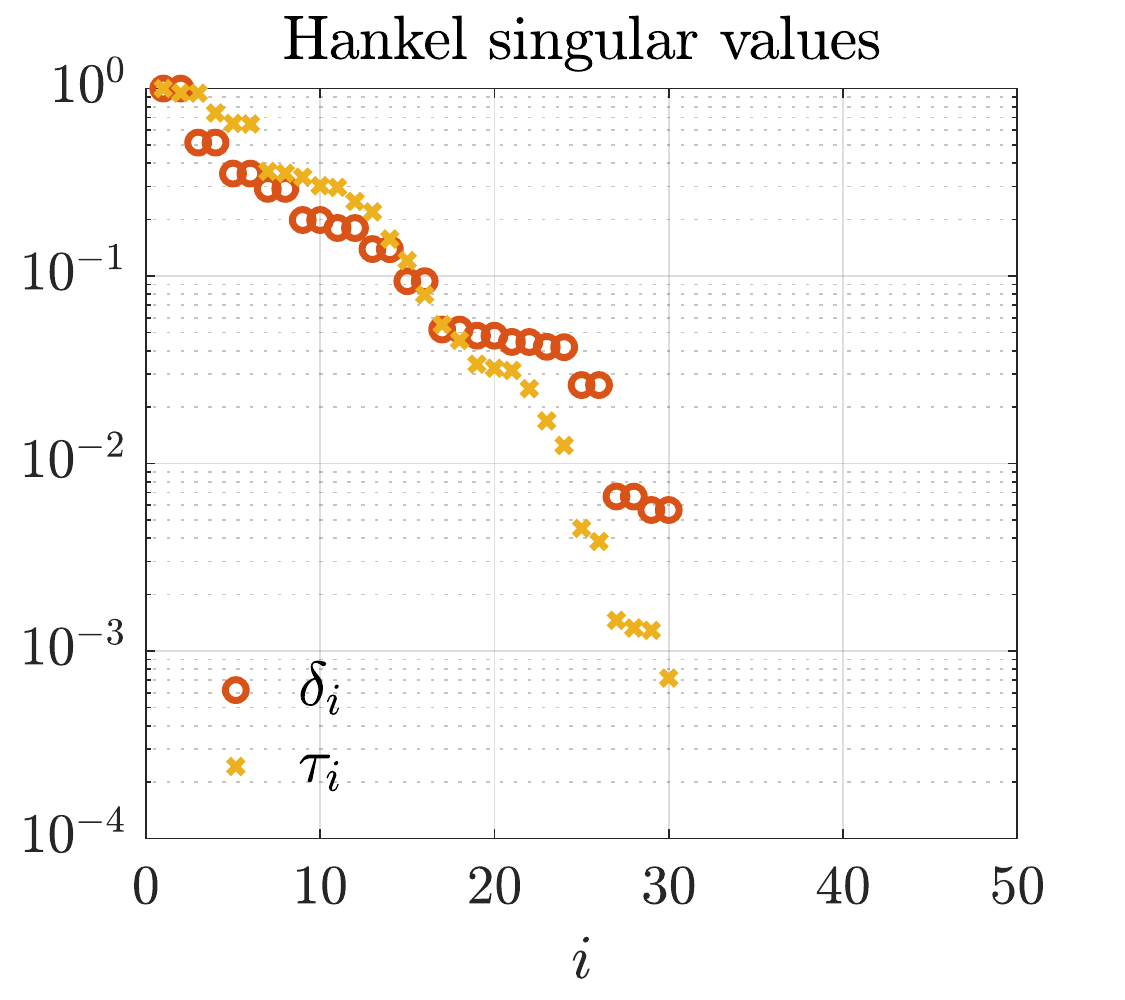}
    \hfill
    \includegraphics[width=0.32\textwidth]{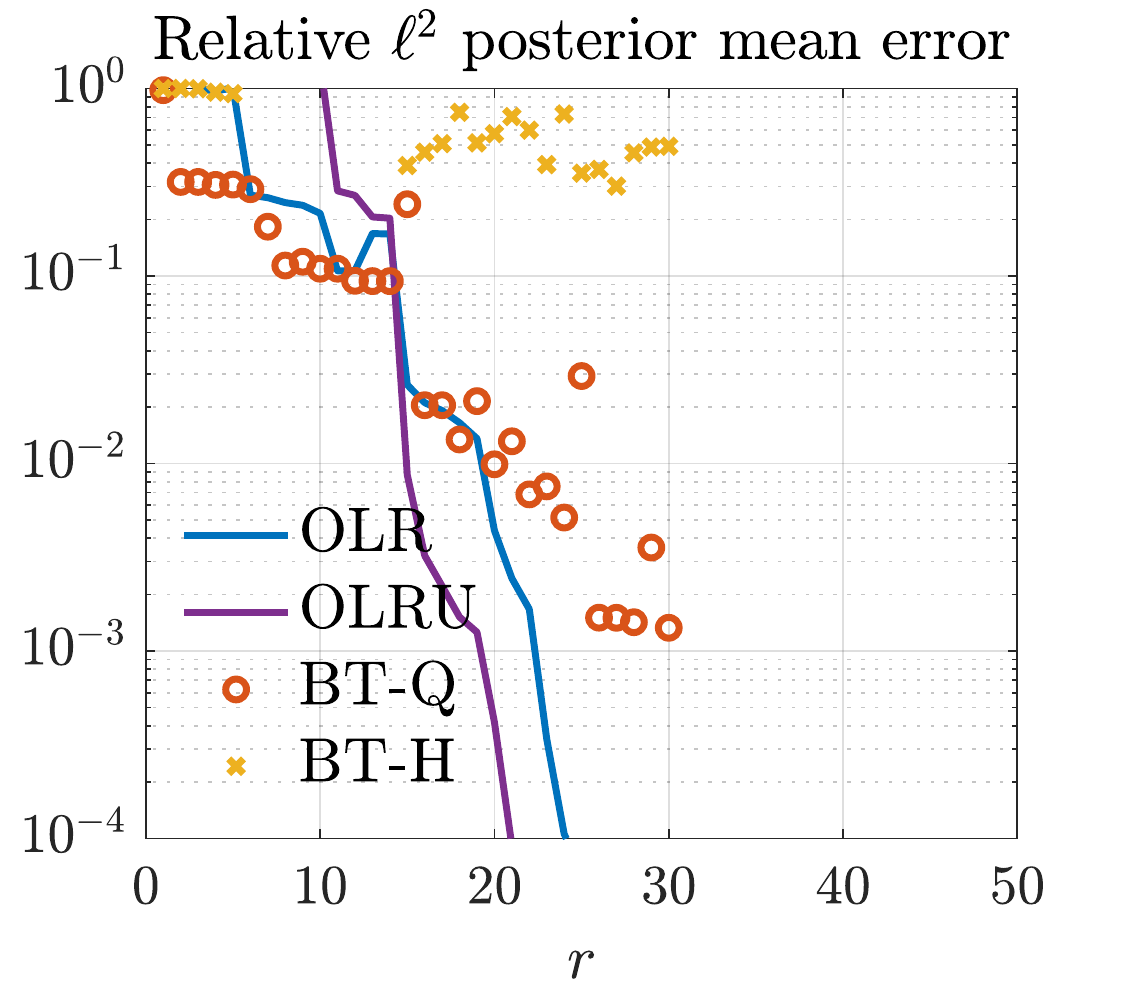}
    \hfill
    \includegraphics[width=0.32\textwidth]{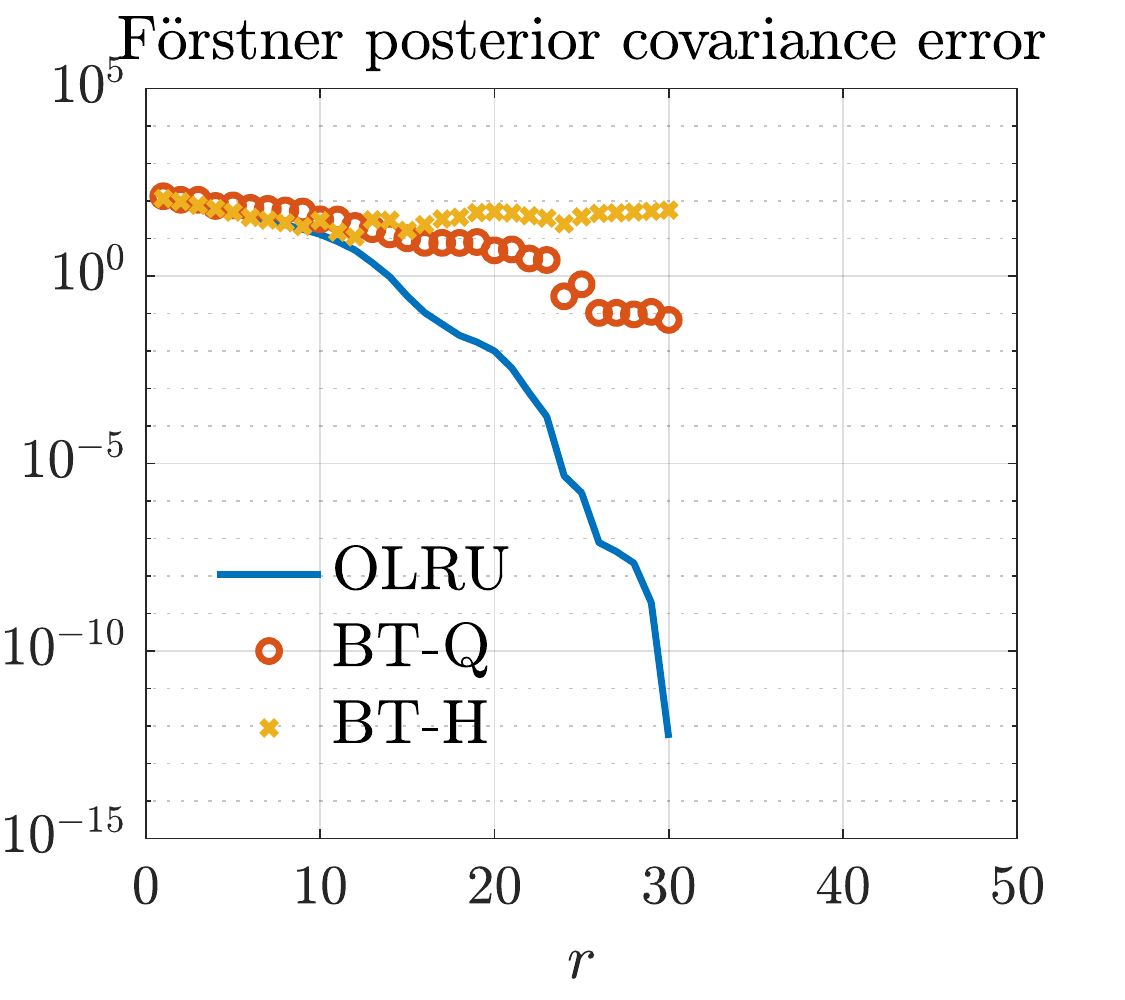}
    \caption{Results for the ISS1R model with measurements spaced $h=1$ apart until $T=10$, for which $\Qm$ and $h\bH$ have a relative error of 53\%. Because $r = 30=\nmeas\nout$, the OLR/U approximations correspond to the full posterior in this case. The quantities plotted are as defined in Figure~\ref{fig:heat good}.
    }
    \label{fig:iss bad}
\end{figure}

\subsection{Discussion}\label{ssec: numerical discussion}
In both test problems, when the observation model has 
a large number of observations, $n$, with high temporal resolution $h \bH$ is a close approximation of $\Qm$ in the Frobenius norm (\Cref{fig:heat good,fig:iss good}). In this case we observe good agreement in $\tau_i$ and $\delta_i$, the square roots of the generalized eigenvalues of the pencils $(\bH,\gprior^{-1})$ and $(\Qm,\gprior^{-1})$, respectively. The posterior covariance approximations using both BT approaches are optimal for low-to-middling values of $r$, indicating agreement not only in the generalized eigenvalues but also in the leading generalized eigendirections. At higher $r$, the balanced truncation posterior covariance approximations diverge from the OLRU approximation, indicating that the trailing eigendirections do not match as well, %perfectly,
but the balanced truncation approximations exhibit respectable errors nevertheless.
In contrast, when the observation model has smaller $n$, with lower temporal resolution of observations (\Cref{fig:heat bad,fig:iss bad}) and shorter observation intervals, we observe poor agreement in $\delta_i$ and $\tau_i$ and significantly sub-optimal posterior covariance approximations from balanced truncation approaches for all but the smallest $r$-values. This indicates generally poor agreement in the generalized eigenpairs of $(\Qm,\gprior^{-1})$ and $(\bH,\gprior^{-1})$, in all but the leading few directions.

Across all test cases, the BT-Q mean error exhibits similar decay to that of the square root of the generalized pencil eigenvalues $\delta_i$ (these are the Hankel singular values of the system~\eqref{eq: modified LTI system}). We note that while it may be possible to bound the posterior mean error using system-theoretic results, the error bound of~\Cref{lem: BT Goodness} does not automatically apply to the posterior mean. Showing that this or a similar bound holds is a direction for future work. In contrast, while the BT-H mean error performs similarly to the BT-Q mean error when $h \bH$ and $\Qm$ are close; in cases where $h\bH$ is far from $\Qm$ the BT-H error is higher, significantly so in the ISS1R case of \Cref{fig:iss bad}. We note that while both the OLR/U mean approximations can achieve lower errors than the balanced truncation approach, the OLR/U mean approximations require evolving the full adjoint dynamics in the dynamical system setting, while the balanced truncation approaches provide a low-dimensional adjoint to evolve. Finally, we observe that in some cases (particularly in Figure \ref{fig:iss bad}) the BT-Q mean error is smaller than that of the OLR approach. This does not contradict the optimality results of \cite{spantini_optimal_2015}, as the OLR mean is optimal in the Bayes risk, not the $\ell^2$-norm in which we compare posterior mean approximations.

Our results show that our proposed approach of~\Cref{ssec: method} provides an efficient low-dimensional reduced model that can be used to compute posterior mean and covariance approximations that are near-optimal when the observation model allows the $(\bH,\gprior^{-1})$ eigenpairs to closely approximate the $(\Qm,\gprior^{-1})$ eigenpairs. However, our results show that even when these eigenpairs are far apart, the BT-Q approach of \Cref{ssec: method} provides approximations with reasonably low (albeit sub-optimal) error.

\section{Conclusions}\label{sec: conclusions}
We consider the linear Gaussian Bayesian inverse problem for inferring the initial condition of a high-dimensional linear dynamical system from noisy measurements taken at $t>0$. In this setting, we propose a new model reduction method that exploits natural connections between the system-theoretic method of balanced truncation and the optimal posterior approximation work of~\cite{spantini_optimal_2015}. To adapt balanced truncation to the inference setting, we define new Gramians for the inference context. We introduce a notion of \emph{prior compatibility} that allows the prior covariance to be interpreted as an infinite reachability Gramian, and we propose two ways of constructing a compatible prior: (i) as the stationary covariance of a system that has been spun-up for $t<0$ by white noise and (ii) via a modification to an incompatible prior. We also define a \emph{noisy observability Gramian} which measures observability energy in a way that accounts for the additive noise in the inference measurements and show that this Gramian coincides with the (scaled) Fisher information matrix in two limits: (i) the limit of continuous observations and (ii) the expectation of countably infinite Poisson-distributed measurement times.

Our proposed approach uses the prior covariance and noisy observability Gramian as inference Gramians in the method of balanced truncation to obtain a reduced model that inherits system-theoretic stability and error guarantees. Additionally, when the dominant generalized eigenpairs of the pencil defined by these inference Gramians match those of the pencil defined by the Fisher information matrix and the prior covariance, our reduced model can additionally recover the optimal posterior covariance approximation of~\cite{spantini_optimal_2015}. Our model yields a reduced dynamical system which can be cheaply evolved to approximate the full system: in this way, our approach extends the work of~\cite{spantini_optimal_2015}, which does not automatically yield a reduced dynamical system. However, the posterior mean approximations considered in \cite{spantini_optimal_2015} and in our approach are different in general. 

Numerical experiments on two benchmark problems from the model reduction community demonstrate that our proposed approach yields near-optimal posterior mean \emph{and} covariance approximations when the noisy observability Gramian is close to the scaled Fisher information matrix.  We demonstrate theoretically and by numerical examples that these matrices are close when the observations are of suitable frequency and duration. When these two matrices differ significantly, the proposed approach yields sub-optimal posterior mean and covariance approximations that nevertheless exhibit low errors that may be acceptable in some applications. 

The connections between linear Gaussian Bayesian inference problems and model reduction via balanced truncation which we have established here are intended as a potential means for transferring perspectives between the two research communities.
Some open questions about the proposed method are directions for future work, including: how to effectively leverage the proposed reduced model for computational gains in data assimilation applications; analysis of the posterior mean approximation; and extension to nonlinear settings. Another natural extension of the proposed approach is to Bayesian inference of the initial state of a system driven stochastically for all time, including through the observation period, which is observed at possibly countably infinite and possibly random times. See \Cref{app:stochalltime} for a discussion of this problem.

\appendix

\section{Proof of \Cref{prop:statcov}}\label{app:proof}
\begin{proof}
The equation~\eqref{eq: white noise dyn sys} can be solved by integrating factors in the standard way because the stochastic product rule applied to $\difd(\expe^{\bA t}x)$ produces no corrections:  
\begin{equation*}
    \bx(t)=\int_{-\infty}^t \expe^{\bA(t-s)} \bB \, \difd \bW (s)
\end{equation*}
With the Ito isometry, which is equivalent to the formal rule $ \E[\difd \bW (s) \difd \bW(s')] = I \delta (s-s') \,  \difd s \, \difd s'$, we compute:
\begin{align}
    \E[\bx(t)\bx^\top(t)] &=
\int_{-\infty}^t \expe^{\bA(t-s)} \bB (\expe^{\bA(t-s)}\bB)^\top \, \difd s
= \int_{-\infty}^t
\expe^{\bA(t-s)} \bB \bB^\top  \expe^{\bA^\top(t-s)} \, \difd s \nonumber \\
&= \int_0^\infty
\expe^{\bA s} \bB \bB^\top  \expe^{\bA^\top s} \, \difd s,
\end{align}
which also shows that $\difd \E[\bx \bx^\top] = 0$. The method of moments gives:
\begin{align} 
    d(\bx\bx^\top) &= (\bA\bx \,\difd t + \bB\, \difd \bW(t))x^\top + \bx( \bx^\top \bA^\top + \difd\bW^\top(t) \bB^\top) + \bB\,\difd\bW(t) \difd\bW(t)^\top \bB^\top
    \nonumber \\
    &= \bA \bx\bx^\top \, \difd t + \bx\bx^\top \bA^\top \, \difd t + \bB\,\difd\bW(t) \bx^\top + \bx \,\difd\bW(t)^\top \bB^\top + \bB\bB^\top \, \difd t \label{eq:statcov_MMalternative}
\end{align}
  Taking expectations and using $\E[\difd\bW x]=0$, we recover the desired Lyapunov equation \eqref{eq:statcov_MMContrGram}:
\begin{align*}
    \difd\E[\bx\bx^\top] = 0 = (\bA \E[\bx\bx^\top] + \E[\bx\bx^\top] \bA^\top + \bB\bB^\top) \, \difd t.  
\end{align*}
\hfill$\Box$
\end{proof}

\section{Optimality Property of Mahalanobis Distance for Gaussian Distributions}
\label{sec:chrishomework}
{
\begin{proposition}
The neighborhood $  \Mset \equiv \{\bz \in \mathbb{R}^d:  D\left(\bz, \cN( \bmu,\Cov)\right) < \delta \}$  defined by the Mahalanobis distance~\eqref{eq: mahalanobis distance} satisfies the following optimality condition:
\begin{equation*}
\Mset = \argmax_{S \in \mathcal{O}: \lambda(S)=\lambda(\Mset)} \Pm(S).
\end{equation*}
where $ \Pm $ denotes the probability measure on $ \mathbb{R}^d$
corresponding to a multivariate Gaussian distribution with  mean $ \bmu \in \mathbb{R}^d$ and covariance matrix $ \Cov \in \mathbb{R}^{d\times d}$,   $\mathcal{O} $ is the collection of open sets in $ \mathbb{R}^d$ and $ \lambda(S)$ denotes the Lebesgue measure of set $S$.
\end{proposition}}

\begin{proof}
{$ M$ can be verified by definition to correspond to an ellipsoidal level set of the multivariate Gaussian density, with the probability density strictly greater on the interior of $M $ than on the complement.  Consequently, for any $ S \in \mathcal{O}$, the probability density on $ S \setminus M$ is strictly greater than the probability density on $ M \setminus S$.  If $ M$ and $ S $ have the same Lebesgue measure, then so do $ S \setminus M$ and $M \setminus S$.  Thus $ P(S \setminus M) \leq P(M\setminus S)$, with equality holding only if $ P(S \setminus M)=0 $, which can only happen if $ S=M$ since both sets are open.} \hfill$\Box$
\end{proof}

\section{Stochastic Forcing Continuing Through Positive Times}
\label{app:stochalltime}

We consider here how the connection between balanced truncation and the optimal low-rank posterior update procedure described in~\cite{spantini_optimal_2015} is affected if we were to consider the more natural case of a stochastic linear dynamical system driving consistently for both positive and negative times:
\begin{align}
    \difd \bx =
        \bA \bx \, \difd t + \bB \,\difd \bW(t), \qquad -\infty < t < \infty.
    \label{eq: white noise dyn sys all}
\end{align}

Now the observations involve noise not only from the measurement process, but from the dynamical noise at previous times.  We can write in place of Eq.~\eqref{eq: measurements}:
\begin{align*}
\bm_i &= \bC \expe^{\bA t_i} \bxin + \bXi_i
, \\
\end{align*}
where
\begin{equation*}
    \bXi_i=
\int_0^{t_i} \bC \expe^{\bA (t_i-s)} \bB \, d \bW (s) + \beps_i.
\end{equation*}
In the expression of the  abstract linear Bayesian inference problem:
\begin{equation*}
\bm = \bG \bxin + \bXi,
\end{equation*}
the noisy component $ \bXi \in \R^{\nobs}$ of the measurements is now correlated across observation times.  Thus, $\bXi \sim \cN(0,\gobsdyn)$ where $\gobsdyn\in\R^{\nobs\times \nobs}$ has block structure:
\begin{align}
    \gobs = \begin{bmatrix}
        \bGamma_{\beps_1,\beps_1} & \bGamma_{\beps_1,\beps_2} & \cdots & \bGamma_{\beps_1,\beps_{\nmeas}} \\
         \bGamma_{\beps_2,\beps_1} & \bGamma_{\beps_2,\beps_2} & \cdots & \bGamma_{\beps_2,\beps_{\nmeas}} \\
        \vdots & \vdots & \cdots & \vdots \\
         \bGamma_{\beps_{\nmeas},\beps_1} & \bGamma_{\beps_{\nmeas},\beps_2} & \cdots & \bGamma_{\beps_{\nmeas},\beps_{\nmeas}} 
    \end{bmatrix}.
    \label{shadeq: block measurement covariance}
\end{align}
with covariances of measurements at different times given explicitly by;
\begin{equation*}
    \bGamma_{\beps_i,\beps_j} = \gmi \delta_{ij}+ \int_{|t_i-t_j|}^{\max(t_i,t_j)}  \bC\expe^{\bA s} BB^\top \expe^{\bA^\top s}\bC^\top \, \difd s
\end{equation*}
with Kronecker delta function $ \delta_{ij}$.  The reason for the correlations is the influence of the common past driving noise.  These correlations are indeed mitigated if the observation times are spaced sufficiently far apart so the minimal eigenvalue of $ -A |t_i-t_j|$ becomes large, but this is not a terribly interesting case as it would imply also weak dependence of the observations on the state $ \bxin$.  The correlation in the effective measurement noise would lead to a very different expression for the Fisher information matrix in~Eq.~\eqref{eq: fisher info sum}), as $ \gobs^{-1}$ would no longer be block diagonal, and we would thereby lose interpretability of the Fisher information matrix as a discretized version of an observability Gramian.

\section{Modification of covariance to induce prior-compatibility}
\label{app:ModPrior}
This appendix complements the discussion in \Cref{sssec: compatibility modification}: We provide below a Matlab routine that, given a system matrix \texttt{A} and a prior covariance \texttt{Gamma0}, returns the upper-triangular Cholesky factor \texttt{R\_Gam} of a modified prior that is compatible with the state dynamics of \texttt{A}.
Our \textsc{matlab} function listed below utilizes the \texttt{lyapchol} routine from the Control System Toolbox to solve \eqref{eq:Deltalyap} for the Cholesky factor \texttt{E} of the prior modification $\Delta$.

{\scriptsize
\begin{verbatim}
function R_Gam=minmodPriorCompat(Gamma0,A)
% Find a minimal residual modification producing 
%   a modified Gamma0 that is prior-compatible with A
% Input: Gamma0 - initial covariance matrix (n x n positive def matrix)
%         A - dynamic system matrix (n x n strictly stable matrix)
% Returns:  R_Gam - the Cholesky factor of a modification of Gamma0 
%   making it compatible with A:  Gamma1 = R_Gam'*R_Gam = Gamma0 + E'*E  
%   such that  A*Gamma1+Gamma1*A' is the *nearest* negative semidefinite 
%   matrix to A*Gamma0+Gamma0*A' - Note that E'*E  is *not* the minimal 
%   perturbation to Gamma0 required for prior-compatibility.
%
% (Requires lyapchol from control system toolbox). 
% 
R0=chol(Gamma0); 
M0=A*Gamma0+Gamma0*A'; 
[V,D]=eig(M0); eigval=real(diag(D));
flagPos=(eigval>0); 
% If Gamma0 is prior-compatible then flagPos is empty
if isempty(flagPos)
    R_Gam=R0;
else
    d_pos=sqrt(eigval(flagPos));
    V_pos=V(:,flagPos);
    E = lyapchol(A,V_pos*diag(d_pos));
    [~,R_Gam]=qr([R0;E],0);
end    
\end{verbatim}
}

\bibliographystyle{spmpsci}      
\bibliography{refs}
\end{document}